%% file: main.tex
\documentclass[11pt]{article}
\usepackage[utf8]{inputenc}
\usepackage[a4paper,left = 1.25in,right = 1.25in, top = 1.25in, bottom = 1.25in]{geometry}
\usepackage{amsmath}
\usepackage{amsthm}
\usepackage{bm}
\usepackage{cancel}
\usepackage{amssymb}
\usepackage{dsfont}
\usepackage{commath}
\usepackage{mathtools}
\usepackage{tikz,tikz-cd}
\usepackage{float}
\usepackage{cases}
\usepackage{url}

\usepackage{enumitem,linegoal}

\usepackage[numbers]{natbib}

\input{Math_Operators}

\usepackage{subcaption}
\usepackage{graphicx}
\usepackage{color}
\usepackage{algorithm}

\usepackage{xcolor}
\usepackage{hyperref}
\hypersetup{colorlinks,linkcolor=blue,citecolor=blue,urlcolor = magenta}
\usepackage[capitalise]{cleveref}

\tikzset{
    myarr/.style={-stealth, thick},
    mylab/.style={font=\small},
}

\usepackage{fancyhdr}
\usepackage{authblk}
\title{\vspace{-10mm} 
Pricing with Passion: \\ 
The Local Occupied Volatility (LOV) Model\\[0.5em]

\vspace{-3mm}}

\date{\today}

\author{Valentin Tissot-Daguette\footnote{Email: {\tt vtissotdague@bloomberg.net}   }}%

\vspace{-2mm}
\affil
{\footnotesize Quantitative Research, Office of the  CTO, Bloomberg \vspace{-2mm} }

\begin{document}
\maketitle
\vspace{-1cm}

\input{Abstract}

\setcounter{tocdepth}{2}
\tableofcontents

\input{Intro}

\input{LOV}

\input{Calibration}

\input{Conclusion}

\bibliographystyle{abbrvnat}
\bibliography{main.bib}
\end{document}

%% file: Math_Operators.tex
\def\be{\begin{equation}}
\def\ee{\end{equation}}

\def\bea{\begin{align}}
\def\eea{\end{align}}

\def\bea*{\begin{align*}}
\def\eea*{\end{align*}}
\def\mo{\mathfrak{o}}
\theoremstyle{plain}

\newtheorem{proposition}{Proposition}

\theoremstyle{definition}

\newtheorem{example}{Example}
\newtheorem{remark}{Remark}

\DeclareMathOperator{\E}{\mathbb{E}}

\DeclareMathOperator{\N}{\mathbb{N}}

\DeclareMathOperator{\Pb}{\mathbb{P}}
\DeclareMathOperator{\Q}{\mathbb{Q}}
\DeclareMathOperator{\R}{\mathbb{R}}

\DeclareMathOperator{\calB}{\mathcal{B}}
\DeclareMathOperator{\calC}{\mathcal{C}}

\DeclareMathOperator{\calL}{\mathcal{L}}
\DeclareMathOperator{\calM}{\mathcal{M}}
\DeclareMathOperator{\calN}{\mathcal{N}}
\DeclareMathOperator{\calO}{\mathcal{O}}

\DeclareMathOperator{\frakF}{\mathfrak{F}}


%% file: Abstract.tex
\vspace{0mm}
\begin{abstract}


We  introduce 
the Local Occupied Volatility (LOV) model that sits  between Dupire’s local volatility and fully path-dependent dynamics.
By design, the LOV model ensures automatic calibration to European vanilla options, 
while offering the flexibility to capture stylized  facts of volatility or fit additional instruments. 
This is achieved by tuning the occupation sensitivity function that quantifies the effect of path-dependent shocks  on volatility. We validate the model through the joint American-European calibration of  options chain on non-dividend paying stocks. 


\end{abstract}
\vspace{2mm}


 \vspace{3mm}


\vspace{3mm}

\textbf{Acknowledgments.} 
I would like to thank Bruno Dupire,  Julien Guyon, Mykhaylo Shkolnikov, and Thibault Jeannin for their stimulating comments.

%% file: Intro.tex
\section{Introduction}

Path-dependent volatility models \cite{HobsonRogers,Guyon2014,GuyonSlides,GuyonLekeufack} have gained in popularity in recent years, notably for their flexibility to capture stylized facts of volatility, fit a variety of liquid contracts, while preserving  market completeness. 
Herein, we illustrate the need for path dependence in volatility through the following practical, yet  overlooked, example.

Consider the calibration of  volatility  to  vanilla options on a U.S. non-dividend paying stock such as Amazon.com Inc. (AMZN). Assume constant, positive interest rate $r  >0$ throughout. 
Although U.S. single name options are American-style, due to the absence of dividends,  the prices of call options  coincide with their European counterparts \cite{Merton73}.  If $c^{E}(K,T),c^{A}(K,T)$ denote respectively the mid-price of the $(K,T)$ European and American call option then  
$$c^{A}(K,T)  = c^{E}(K,T), \quad \forall K,T.$$ 
Classically, the local volatility (LV)  model \cite{DupireLV} provides   perfect calibration to European vanilla options; if under $\Q_{\text{loc}}$, the underlying evolves according to 
\begin{equation}\label{eq:LV}
    \frac{dX_t}{X_t} = rdt + \sigma_{\text{loc}}(t,X_t) dW_t,
\end{equation}
where $\sigma_{\text{loc}}$ is obtained from Dupire's formula, 
then 
$  c^E(K,T) = \E^{\Q_{\text{loc}}}[e^{-rT} (X_{T}- K)^+]$ holds for all $K\ge 0$. 
On the other hand, the put options do carry an early exercise premium, and there is no guarantee that the local volatility model generates exactly the observed put option quotes. That is, there typically exist $(K,T)-$put options such that 
$$p^A(K,T) \  \ne \  \sup \hspace{0.2mm} \big\{ \E^{\Q_{\text{loc}}}[e^{-r\tau} (K-X_{\tau})^+] \ : \ \tau \le T \; \text{ stopping time} \big \}, $$
where the right-hand side is the  price of the American put  under $\Q_{\text{loc}}$. 
Consequently, one needs to look beyond Markovian models to jointly fit single name call and  put options. In other words, 
volatility \textit{must be} path-dependent. 

The aim of  this work is to   explore  a specific type of path dependence to model volatility, generated by the   \textit{occupation times} 
$$\calO_t(A) = \int_0^t \mathds{1}_{A}(X_s) e^{\kappa s}ds, \quad A\subseteq \R_+,$$
for some exponential decay parameter $\kappa \ge 0$. 
The use of occupation times is motivated by recent studies on  \textit{occupied stochastic differential equations} 
\cite{TissotOP,TissotThesis,TissotZhang}, where $\calO_t$ notably  appear in the diffusion coefficient, i.e.,  volatility. The Markovian model \eqref{eq:LV} is thus extended to the path-dependent dynamics,
\begin{equation}\label{eq:OV}
    \frac{dX_t}{X_t} = rdt + \sigma_t dW_t, \quad \sigma_t = \sigma(\calO_t,X_t).
\end{equation}
As argued in \cref{sec:nodividends}, careful choice of the occupied volatility 
can solve the aforementioned joint American-European calibration problem.
First, the volatility process in \eqref{eq:OV}  must  satisfy the \textit{smile calibration condition},
\begin{equation}\label{eq:smileIntro}
    \E^{\Q}[\sigma^2_t \ | \ X_t = x] = \sigma_{\text{loc}}^2(t,x), \quad \forall \  t,x,
\end{equation}
where  $\Q$ is the risk neutral measure associated with \eqref{eq:OV}. From  \cite{dupire2004unified,gyongy}, identity \eqref{eq:smileIntro} is equivalent to a perfect match of all European call option prices under $\Q$. 
The smile calibration condition is guaranteed when 
applying a first order  \textit{Markovian expansion}  (see \cref{sec:MarkovianExpansion}) to the volatility functional in \eqref{eq:OV}. We then postulate the local occupied volatility (LOV) model, 
\begin{align}
  \sigma^2_t = \sigma_{\text{loc}}^2(t,X_t) + 
\gamma_t \int_{\R_+}\ell(t,X_t,x) (\mathcal{O}_t-\hat{\mathcal{O}}_t)(dx),      \label{eq:LOVOIntro}
\end{align}
with $\gamma_t = \calO_t(\R)^{-1}$ and  the projected measure $\hat{\mathcal{O}}_t = \E^{\Q}[\calO_t  | X_t ]$. Indeed, the smile calibration condition \eqref{eq:smileIntro} holds, since the rightmost term in \eqref{eq:LOVOIntro} has been  centered.  
Subsequently, the \textit{sensitivity function} $\ell(\cdot)$ can be tuned to match the American put option quotes. The LOV model is still useful when all vanilla options are European, e.g. if $X$ is the price of an index,  in which case the sensitivity function is used to create more realistic spot-vol dynamics than in a plain LV model, see \cref{subsubsec:LOV}. 


\begin{figure}[t]
    \centering
        \caption{Vanilla options on Amazon Inc. (AMZN). The put options carry an early exercise premium, while the calls do not.  }
    \includegraphics[width=0.75\linewidth]{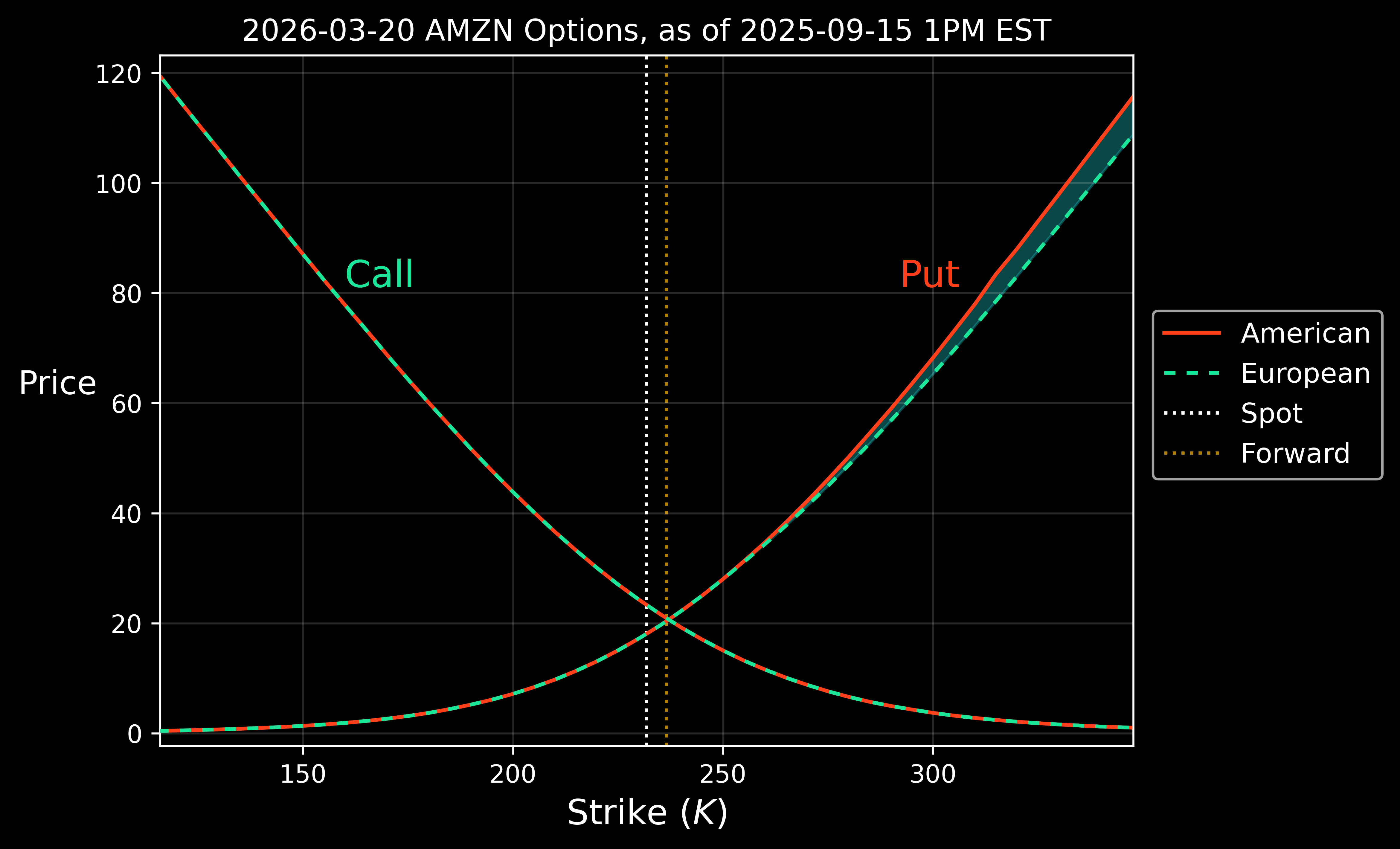}

    \label{fig:amznOPtions}
\end{figure}

%% file: LOV.tex
 \section{Occupied Volatility}
 
Consider  the calendar time occupation flow  $\calO = \int_0^{\cdot} \delta_{X_s} ds$ \cite{TissotOP} and  assume that the asset price process $X$ follows the  dynamics, 
\begin{equation}\label{eq:OVM}
  \frac{dX_t}{X_t} = (r-q)dt+ 
    \sigma(\calO_t,X_t)dW_t, 
\end{equation} 
with interest rate $r\in \R$ and dividend yield $q\ge 0$, both assumed to be constant for simplicity. 
The diffusion coefficient $\sigma$ shall be termed    \textit{occupied volatility}. 
A special case is Dupire's local volatility model \cite{DupireLV}. Indeed, if $\sigma_{\text{loc}}$ 
denotes the  local volatility function, then 
\begin{equation}
    \sigma(\mo,x) := \sigma_{\text{loc}}(\mo(\R),x) \; \Longrightarrow \; \sigma(\calO_t,X_t)  = \sigma_{\text{loc}}(t,X_t).
\end{equation}
Hence \eqref{eq:OVM}  gives a sensible  generalization of the local volatility model  as long as  $\sigma$  satisfies  the 
\textit{smile calibration condition} \cite{BrunickShreve,dupire2004unified,gyongy},   
\begin{equation}\label{eq:smileCalibration}
     \E^{\Q} [\sigma^2(\calO_t,X_t)\ |\ X_t] = \sigma_{\text{loc}}^2(t,X_t). 
\end{equation}
 ensuring that vanilla call and put options are correctly priced. 
On the other hand, occupied volatility is a particular example of  path-dependent volatility  \cite{BrunickShreve,Guyon2014,GuyonSlides,GuyonLekeufack,HobsonRogers} and thus  enjoys similar benefits. Notably,  as the occupation flow is  endogenous,  the model remains complete, i.e.,   any contingent claim  can be perfectly replicated by dynamically trading   the underlying.  

From an econometric perspective, 
 calendar time may not be the most realistic clock to model volatility, due to  the chronology invariance of $\calO$. 
As an illustration, fix a one-year horizon with current time $t = 1$, and consider two scenarios $\omega^1$, $\omega^2$ such that $X_1(\omega^1) = X_1(\omega^2)$ and $\calO_1(\omega^1) = \calO_1(\omega^2)$. Then $\omega^1$,  $\omega^2$ generate the same occupied volatility level. However,  suppose that in the first (respectively second) scenario, a financial crash occurred a week (resp. a year) ago, while the market has been bullish otherwise. One would then expect  much  higher volatility  in the first scenario.

\begin{figure}[t]
\centering
\caption{Occupation measure using calendar time ($\calO_t$, left)  compared with exponential time ($\calO$, right). As can be seen, $\calO$ puts more emphasis on recent levels of the path.  }

\begin{subfigure}[b]{0.47\textwidth}
     \centering
     \caption{Calendar time. Preferred under $\Q$.}
\includegraphics[height=2in,width=2.6in]{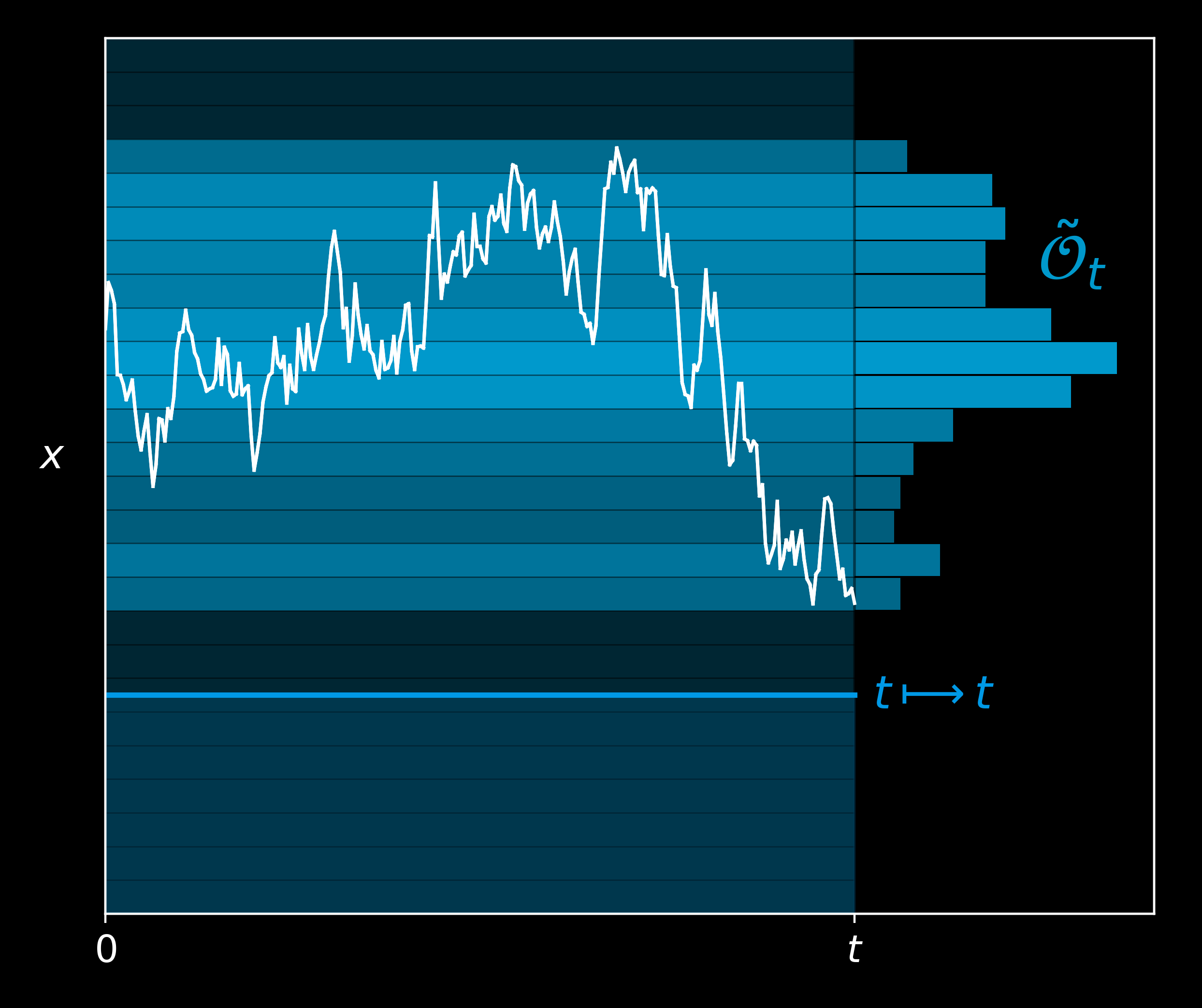}
\end{subfigure}
\begin{subfigure}[b]{0.47\textwidth}
     \centering
     \caption{Exponential time. Preferred under $\Pb$.}
\includegraphics[height=2in,width=2.6in]{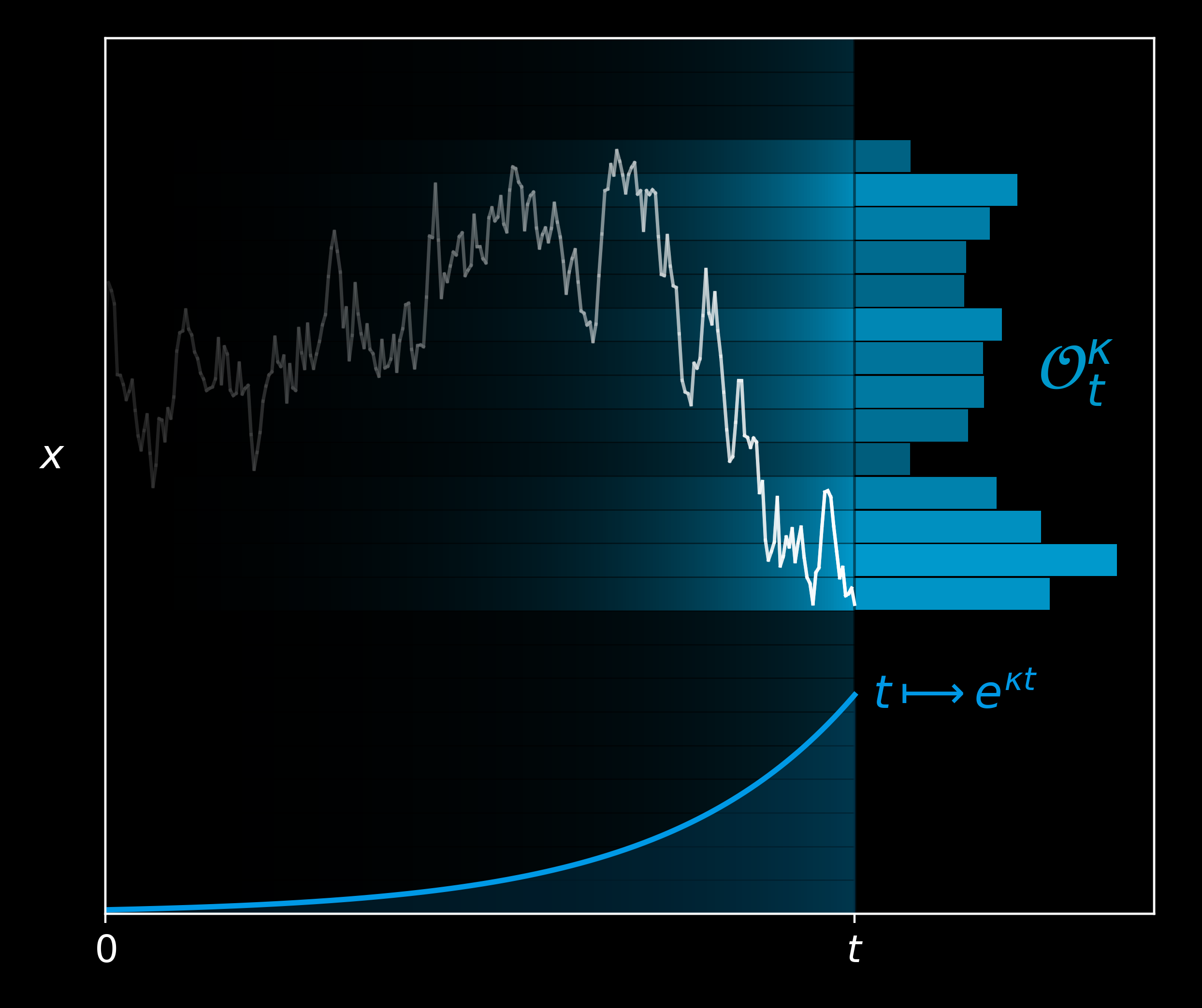}
 \end{subfigure}
\vspace{5mm}
 \label{fig:OSDEVol}
 \end{figure}

The above example  
calls for a clock  that captures a memory decay in past information. 
Among others, this is achieved through the \textit{exponential time} occupation flow, 
\begin{equation}\label{eq:occDecay}  d\calO_t = \delta_{X_t}e^{\kappa t}dt,  \quad \kappa \ge 0. 
\end{equation}
The calendar time occupation flow is recovered by setting $\kappa = 0$.  
Note that calendar time is still relevant for the pricing/hedging of exotic options due to the unified framework  it offers; see Section 4 in \cite{TissotOP}. 
Hence in general, 
calendar time is preferred under $\Q$ (pricing measure), while exponential time is preferred under $\Pb$ (historical measure).


The volatility now takes the exponential time occupation flow as input, leading to 
\begin{equation}\label{eq:OVMExp}
  \frac{dX_t}{X_t} = (r-q)dt+  
    \sigma(\calO_t,X_t)dW_t. 
\end{equation} 
By increasing the 
exponential decay parameter $\kappa$, the occupation flow puts more emphasis on the present. 
This adjustment makes volatility more reactive to recent shocks,  and  allow the model to  replicate stylized facts such as volatility clustering. 


\begin{example}\label{ex:occVol} \textnormal{(Guyon's toy volatility model)}  
Consider the occupied volatility    
\begin{equation}\label{eq:volOSDE} 
   \sigma(\calO_t,X_t) = \Sigma( \Upsilon(\calO_t,X_t)),  \quad \Sigma(y) = -\frac{\alpha}{\beta} + \gamma y^{-\beta}, 
\end{equation}
with parameters $\alpha,\beta,\gamma>0$, and functional 
\begin{equation}\label{eq:trend}
\Upsilon(\calO_t,X_t) = \frac{X_t}{\frac{1}{\calO_t(\R)}\int_{\R} x \calO_t(dx)}. 
\end{equation}
The ratio \eqref{eq:trend} characterizes the asset's trend by noting that  the barycenter of  $\calO$ coincides with the exponentially weighted moving average (EMA) of the asset,  
\begin{equation}\label{eq:EWMA}
     \text{EMA}_t :=  \frac{1}{1-e^{-\kappa t}} \int_0^t  X_u \kappa e^{-\kappa (t-u)} du = \frac{1}{\calO_t(\mathbb{R})} \int x \calO_t(dx). 
\end{equation} 
Then $dX_t =  \sigma(\calO_t,X_t)X_tdW_t$  is 
precisely Guyon's toy volatility model  \cite{GuyonSlides}, further   analyzed by \citet{JackBonesini}. 
Since $\Sigma$ in \eqref{eq:volOSDE} is a  decreasing function,    the volatility  $\sigma$  is inversely proportional to the asset's trend. 
 This is in line with 
 the \textit{leverage effect}, a stylized fact in equity markets that volatility   increases when the asset price drops.  

 
\end{example}



\subsection{Simulation of Occupied SDEs}\label{app:numOSDE}


The numerical integration of OSDEs inevitably requires a discretization of the occupation flow.  
Focusing on the occupied volatility model 
\eqref{eq:OVMExp}, 
a natural way consists of keeping track of finitely many occupation times  $\calO_{t}(C_1), \ldots, \calO_{t}(C_M)$, where $(C_m)$ forms a partition of $\R_+$. 
Concretely, let $  x_1 < \ldots < x_M $,  $M\in \N$, and introduce the corridors,  
 $$C_m = [x_m - \varepsilon_{m-1}, x_m + \varepsilon_m),  \quad \varepsilon_m = \frac{x_{m+1}-x_m}{2}, \quad (\varepsilon_{0} = x_1, \varepsilon_M = \infty).$$ 
The nodes $(x_m)$ are chosen so as to match the scale of $X_0$ and OSDE coefficients. 
We then replace $\calO_t$  by 
the discrete measure $  \sum_{m=1}^M  \calO_{t}(C_m) \delta_{x_m}$, 
and approximate OSDE \eqref{eq:OVMExp}  by the finite-dimensional system
\begin{equation}\label{eq:finiteOSDE}
    \frac{dX_t}{X_t} = (r-q)dt +  \sigma^M(O_t,X_t)  dW_t, \qquad dO_{t,m} = \mathds{1}_{C_m}(X_t) e^{\kappa t}dt, 
\end{equation}
with $O_t = (\calO_{t}(C_1), \ldots, \calO_{t}(C_M))\in \R^M$. In light of \eqref{eq:finiteOSDE},  $(O,X)$ is an $M+1-$-dimensional Markov process. Moreover, the projected volatility $$\sigma^M:\R^{M+1}\to \R, \qquad \sigma^M(O_t,X_t) = \sigma\Big(\sum\limits_{m=1}^M O_{t,m} \delta_{x_m},X_t\Big),$$
 converges to $\sigma(\calO_t,X_t)$ when $\sigma(\mo,x)$ is weakly continuous in $\mo$. 
 We then generate sample paths from \eqref{eq:finiteOSDE}  as follows using standard discretization techniques \cite{KloedenPlaten,Glasserman2003}.  Given the   time grid $t_n = n\delta t$, $\delta t = \frac{T}{N}$, $N\in \N$,  $X_0 \in \R_+$,  and $O^{\text{\tiny tmp}} = 0 \in \R^M$,  the SDE \eqref{eq:finiteOSDE} is updated  according to 
\begin{align*}
\begin{cases}
     \sigma_{t_{n}}&=   \      \sigma^M(O^{\text{\tiny tmp}},X_{t_n}), \\[0.3em] 
     X_{t_{n+1}} &= \ X_{t_{n}} \exp\{\sigma_{t_n} \sqrt{\delta t}Z_{n} + (r-q- \frac{1}{2} \sigma_{t_n}^2)\delta t \},\\[0.3em]
          O^{\text{\tiny tmp}}_m&\leftarrow \  O^{\text{\tiny tmp}}_m +   e^{\kappa t_n} \delta t, \quad X_{t_{n}} \  \in \ C_m, 
     \end{cases}
\end{align*}
 where $Z_n \overset{i.i.d.}{\sim} \calN(0, 1)$, $n=0,\ldots,N-1$. 
Note that the occupation measure is represented by the temporary variable $O^{\text{\tiny tmp}}$ updated at each iteration.  Indeed, 
 it is not necessary (and costly) to store the occupation times $O_{t_n}$ for each $n$. 

 \begin{remark}
     More generally, one may employ 
 a partition of unity $\{f_m:\R_+ \to [0,1]\}_{m\le M}$, $\sum_{m} f_m \equiv 1$,   leading to the \textit{cylindrical projection}  $\sigma(\calO_t,X_t)  \approx  \sigma^M(f_1 \cdot \calO_t,\ldots,f_M \cdot \calO_t,X_t)$. The  occupation time approximation  is thus recovered by setting $f_m = \mathds{1}_{C_m}$, and  similarly to $(C_m)$ being a partition, the property $\sum_{m} f_m \equiv 1$ ensures that the total mass of  $\calO$ is preserved. A convergence analysis of cylindrical projections of  OSDEs as $M\to \infty$  
is  addressed in the companion  paper \cite{TissotZhang}.
 \end{remark}

\begin{remark}\label{rem:PSDEs}
   The simulation of fully path-dependent volatility models,   $$\frac{dX_t}{X_t} = (r-q)dt + \sigma(t,\omega)dW_t, $$  
typically  involves discretization of the form 
$\sigma(t_n,\omega) \approx \sigma^N_{n}(X_{t_1}(\omega), \ldots, X_{t_{n}}(\omega))$  for some  function $\sigma^N_n:\R^n\to \R$.  
Note the input dimension of $\sigma^N_n$ grows linearly with $n$,  which becomes  highly problematic when the total number of time steps $N$ is large. 
In contrast, occupied SDEs decouple the dimensionality of the space and time discretization, allowing for efficient simulations 
as shown in \cref{ex:occVol} and \cref{sec:calibration}. 
\end{remark}

\section{Local Occupied Volatility (LOV)}\label{sec:LOV}

\subsection{Markovian Expansion} \label{sec:MarkovianExpansion}
Let us start with an autonomous path-dependent volatility model of the form 
$$\frac{dX_t}{X_t} = (r-q)dt  +  \sigma(F_t,X_t)dW_t,  $$
for  some path-dependent feature $F_t(\omega) = \textnormal{f}(t,\omega) \in \frakF$, $\textnormal{f}:[0,T]\times \Omega_T \to \frakF$. The  \textit{feature space} $\frakF$ may be finite-dimensional, e.g. 
$F_t \in \{\sup_{s\le t}X_s, \int_0^tX_sds\}$  in \citet{BrunickShreve}, 
or infinite-dimensional in occupied volatility models where $\frakF = \calM$; see \cref{subsubsec:LOV}.   
Recall the smile calibration condition 
\begin{equation}\label{eq:smileCalibration}
    \sigma_{\text{loc}}^2(t,X_t) = \E^{\Q}[\sigma^2(F_t,X_t) \ | \ X_t], 
\end{equation}
 and introduce the \textit{projected feature} 
$\hat{F}_t = \E^{\Q}[F_t|X_t] = \hat{f}(t,X_t)$ for some map $ \hat{f}: \R_+^2\to \frakF.$ 
If  $\sigma^2(F,X)$ were linear in $F$, then 
$$ \E^{\Q}[\sigma^2(F_t,X_t)  |  X_t]  =\sigma^2(\E^{\Q}[F_t |  X_t],X_t)   = \sigma^2(\hat{F}_t,X_t),$$
so \eqref{eq:smileCalibration} would imply that  $\sigma_{\text{loc}}(t,X_t) = \sigma(\hat{f}(t,X_t),X_t)$. 
In general, one can linearize $F\mapsto \sigma^2(F,X)$ via a first order \textit{Markovian expansion} around  $\hat{F}_t$, that is
\begin{align}\label{eq:MarkovExpansion}
    \sigma^2(F_t,X_t) \approx \sigma^2(\hat{F}_t,X_t) + D\sigma^2(\hat{F}_t,X_t)\cdot (F_t - \hat{F}_t), 
\end{align}
with Fréchet derivative $D\sigma^2(\cdot,X_t): \frakF \to \frakF^*$ and  pairing  $  \ell \cdot F = \ell(F) \in \R,$ $(\ell,F) \in \frakF^*\times \frakF$. Note that  $D\sigma^2(\hat{F}_t,X_t)$ is the sensitivity of  variance with respect to  non-Markovian shocks. 
In view of \eqref{eq:smileCalibration} and the fact that $\sigma^2(\hat{F}_t,X_t)$, $D\sigma^2(\hat{F}_t,X_t)$ are functions of $(t,X_t)$, we can \textit{postulate}  the following local path-dependent volatility (LPDV) model,  
\begin{align}
    \frac{dX_t}{X_t} &= (r-q)dt  +  \sigma_tdW_t, \label{eq:LPDV0} \\[0.5em]
    \sigma_t^2 &= \sigma_{\text{loc}}^2(t,X_t) + \ell(t,X_t) \cdot (F_t - \hat{F}_t). \label{eq:LPDV}
\end{align}
That is, we have replaced the term  $\sigma^2(\hat{F}_t,X_t)$ in \eqref{eq:MarkovExpansion} with the local variance, and  the  derivative $D\sigma^2(\hat{F}_t,X_t)$ with an arbitrary map $\ell: \R_+^2\to \frakF^*$. Due to   $\E^{\Q}[F_t - \hat{F}_t|  X_t] =0$, the above specification automatically satisfies the smile calibration condition \eqref{eq:smileCalibration} for any reasonable choice of $\ell$, e.g.  such that $\sigma^2_t$ remains positive. Since $\hat{F}_t$ depends on $\Q$, so does $\sigma_t$ in \eqref{eq:LPDV}, leading to McKean type dynamics. 


\begin{remark}
    The LPDV model   \eqref{eq:LPDV} is reminiscent of local stochastic volatility (LSV) models \cite{GuyonHenryLabordere2012,GuyonHL} where one multiplies $\sigma_{\text{loc}}$ by an exogenous stochastic processes $\alpha_t$ and normalized 
such that \eqref{eq:smileCalibration} holds. The volatility process is thus of the form 
\begin{equation}\label{eq:LSV}
    \sigma_t  = \sigma_{\text{loc}}(t,X_t) \frac{\alpha_t}{\sqrt{\E^{\Q}[\alpha_t^2|  X_t]}}, 
\end{equation}
where the conditional expectation is 
typically estimated using the particle method of \citet{GuyonHenryLabordere2012}. It is a powerful family of models that is consistent with vanilla options and capable of  producing more realistic spot-vol dynamics than the standard local volatility model. 
On the other hand,  local path-dependent volatility  models also   preserve market completeness by explaining future volatility from the price history  directly. Moreover, considering an additive correction term 
 is arguably preferred numerically over the fraction in \eqref{eq:LSV} with  distribution-dependent   denominator. Indeed,  any  error introduced when computing  $\hat{F}_t$  
is controlled in a linear way.
\end{remark}

\subsection{The LOV Model}\label{subsubsec:LOV}

Let us set $F_t = \calO_t$. 
 The feature space is $\frakF = \calM$, and the Fréchet derivative in \eqref{eq:MarkovExpansion} becomes the linear derivative \cite{CarmonaDelarue}, which we denote by $\delta_{\mo}$. 
 Recalling the dual pairing $\phi\cdot \mo = \int_{\R}\phi(x)\mo(dx)$, $(\phi,\mo)\in \calC_b(\R) \times \calM$, we can then rewrite \eqref{eq:MarkovExpansion} as 
 \begin{align}\label{eq:MarkovExpansion2}
    \sigma^2(\calO_t,X_t) \approx \sigma^2(\hat{\calO}_t,X_t) + \int_{\R_+}\delta_{\mo}\sigma^2(\hat{\calO}_t,X_t)(x) \ (\calO_t - \hat{\calO}_t)(dx). 
\end{align}
The \textit{occupation sensitivity} $\delta_{\mo}\sigma^2(\hat{\calO}_t,X_t)(x) = \frac{d}{dh}\sigma^2(\hat{\calO}_t + h\delta_x,X_t)|_{h=0}$ describes the change in  variance when the asset spends more time   at level $x$ than implied by the expected occupation measure $\hat{\calO}_t$. 
We then replace $\sigma^2(\hat{\calO}_t,X_t)$  with the local variance, and rewrite the occupation sensitivity as 
$$\delta_{\mo}\sigma^2(\hat{\calO}_t,X_t)(x) = \gamma_t^\kappa \ell(t,X_t,x), \quad \gamma_t^\kappa = \calO_t^\kappa(\R)^{-1} =
    \frac{\kappa}{e^{\kappa t} - 1}, \; \,  \kappa > 0,  
    \quad \gamma_t^0 = \frac{1}{t}. $$
Altogether, we have  derived the Local Occupied Volatility (LOV) model, 
\begin{align}
      \frac{dX_t}{X_t} &= (r-q)dt+  
    \sigma_tdW_t, \label{eq:LOVX}\\ 
  \sigma^2_t &= \sigma_{\text{loc}}^2(t,X_t) + 
\gamma_t^\kappa \int_{\R_+}\ell(t,X_t,x) (\mathcal{O}_t-\hat{\mathcal{O}}_t)(dx).    \label{eq:LOVO}
\end{align}
The factor $\gamma_t^\kappa $ is introduced so as to maintain a constant  magnitude of $\ell$ over time, and allow  to  interpret the linear term as the average  sensitivity of variance when the realized occupation measure $\mathcal{O}_t$ deviates from the  implied   $\hat{\mathcal{O}}_t$. Similar to \eqref{eq:LPDV}, the occupied variance  depends on $\Q$ through  $\hat{\mathcal{O}}_t$, hence $\sigma_t = \sigma(\calO_t,X_t,\Q)$.   
 
Without careful specification of the sensitivity function ($\ell$), 
the LOV model may produce negative variances because the path-dependent term in \eqref{eq:LOVO} is centered around zero. However, mild conditions on $\ell$ can guarantee $\sigma^2$
  remains positive, as shown next.  
  
\begin{proposition}\label{prop:positivity}
Consider the LOV model \eqref{eq:LOVX}--\eqref{eq:LOVO} with sensitivity function  $\ell$ such that  $|\ell(t,x',x)| < \frac{1}{2}\sigma_{\textnormal{loc}}^2(t,x')$ for all $t,x',x$. Then $\sigma_t(\omega) >0$  for all $t$, $\omega\in \Omega$.  
\end{proposition}
\begin{proof}
Let $c(t,x') = \sup_{x}|\ell(t,x',x)|$. By assumption, $c(t,x')< \frac{1}{2}\sigma_{\textnormal{loc}}^2(t,x')$, hence
    $$  \sigma^2_t \ge \sigma_{\textnormal{loc}}^2(t,X_t) - 
\gamma_t^\kappa \int_{\R}|\ell(t,X_t,x)|  (\mathcal{O}_t+\hat{\mathcal{O}}_t)(dx)  \ge \sigma_{\textnormal{loc}}^2(t,X_t) -2c(t,X_t) >0.  $$
\end{proof}
Under the above assumption on $\ell$, the sensitivity function can nearly double the local variance or push it down  closer to  zero. The path-dependent term in \eqref{eq:LOVO} thus offers a wide
range to adjust the local variance while ensuring positivity.
\begin{remark} \label{rem:multiply}
    Equivalently, the local occupied variance can be reformulated using a multiplicative correction, that is  
\begin{equation}
\sigma_t^2 = \sigma_{\text{loc}}^2(t,X_t) \Big(1  + 
\gamma_t^\kappa \int_{\R_+}\tilde{\ell}(t,X_t,x) (\mathcal{O}_t-\hat{\mathcal{O}}_t)(dx) \Big), \label{eq:LOVMultiply}
\end{equation}
where  $\tilde{\ell}$  represents the path-dependent sensitivity of variance relative to the baseline local variance.
Analogous to \cref{prop:positivity}, the condition $|\tilde{\ell}(t,X_t,x)| < \frac{1}{2}$ ensures the strict positivity of $\sigma^2$.
\end{remark}

 \begin{example}\label{eq:simpleLOV} (one-factor LOV) 
Consider the multiplicative version \eqref{eq:LOVMultiply}, with  $\tilde{\ell}(t,X_t,x) = \beta \mathds{1}_{A}(x)$,  $A\in \calB(\R_+)$, and   $|\beta| < 1/2$. Then, 
$$
\sigma_t^2 = \sigma_{\text{loc}}^2(t,X_t) \Big(1  + \beta \gamma_t\big (\mathcal{O}_t - \hat{\mathcal{O}}_t \big)(A)\Big).$$ 
In view of \cref{prop:positivity},  $\sigma^2 > 0$. Recalling that $\gamma_t = \calO_t(\R_+)^{-1} = \hat{\calO}_t(\R_+)^{-1}$,  the local variance is thus rescaled according to  the  fraction of time   $\gamma_t\mathcal{O}_t(A)$ 
spent by the asset in $A$ compared to what is  expected. Here, the occupied volatility depends only on $t,X_t$, and the occupation time $\calO_t(A)$. As $(\calO(A),X)$ is a two-dimensional Markov process, trajectories from the LOV model can thus be effectively simulated under this simple specification. In  \cref{alg:LOV}, this  would correspond to choosing $M=1$ corridor. 

Examples for $A$ include the upside corridor $A^{+} = [X_0,\infty)$, where a positive parameter $\beta$ would be consistent with the leverage effect.  Alternatively, one can consider a narrow corridor around the money, e.g. $A = [0.9X_0,1.1X_0]$, where high values for $\mathcal{O}_t(A)$ indicate some mean reversion of the recent path toward the initial value. This mean reversion would, in turn,  translate into lower volatilities than $\sigma_{\text{loc}}$ by choosing $\beta < 0$.  
\end{example}

\begin{example} \label{ex:EMA} (EMA LOV) 
Suppose that $\ell(t,X_t,x) = \beta \log x$ for some parameter $\beta\in \R$.
    Then, using \eqref{eq:EWMA}, we can rewrite  the local occupied volatility \eqref{eq:LOVO} as 
    \begin{align*}
      \sigma_t^2
&=\sigma_{\textnormal{loc}}^2(t,X_t) + \beta   (\text{EMA}_t -\widehat{\text{EMA}}_t), \\[1.2em] 
\text{EMA}_t &=   \frac{1}{1-e^{-\kappa t}} \int_0^t  \log X_u \kappa e^{-\kappa (t-u)} du, \qquad \;  
\widehat{\text{EMA}}_t =  \E^{\Q}[\hspace{0.2mm} \text{EMA}_t  \ | \  X_t\hspace{0.2mm}].
 \end{align*}
 The trend $\text{EMA}_t -\widehat{\text{EMA}}_t$ is 
similar to the first offset process of \citet{HobsonRogers}, where the projected  EMA is replaced by the log spot $\log X_t$. Note that choosing $\beta > 0$ makes  the path-dependent correction term consistent with the leverage effect.  
\end{example}

 \begin{figure}[t]
\centering
\caption{Local occupied volatility (yellow dot) compared with local volatility (red). The yellow curve is the local sensitivity function $x\mapsto \ell(t,X_t,x)$, integrated against the spread  $\calO_t-\hat{\calO}_t$ between the realized (blue) and implied  (purple) occupation measure.   }

\begin{subfigure}[b]{0.482\textwidth}
     \centering
\includegraphics[height=2.4in,width=2.86in]{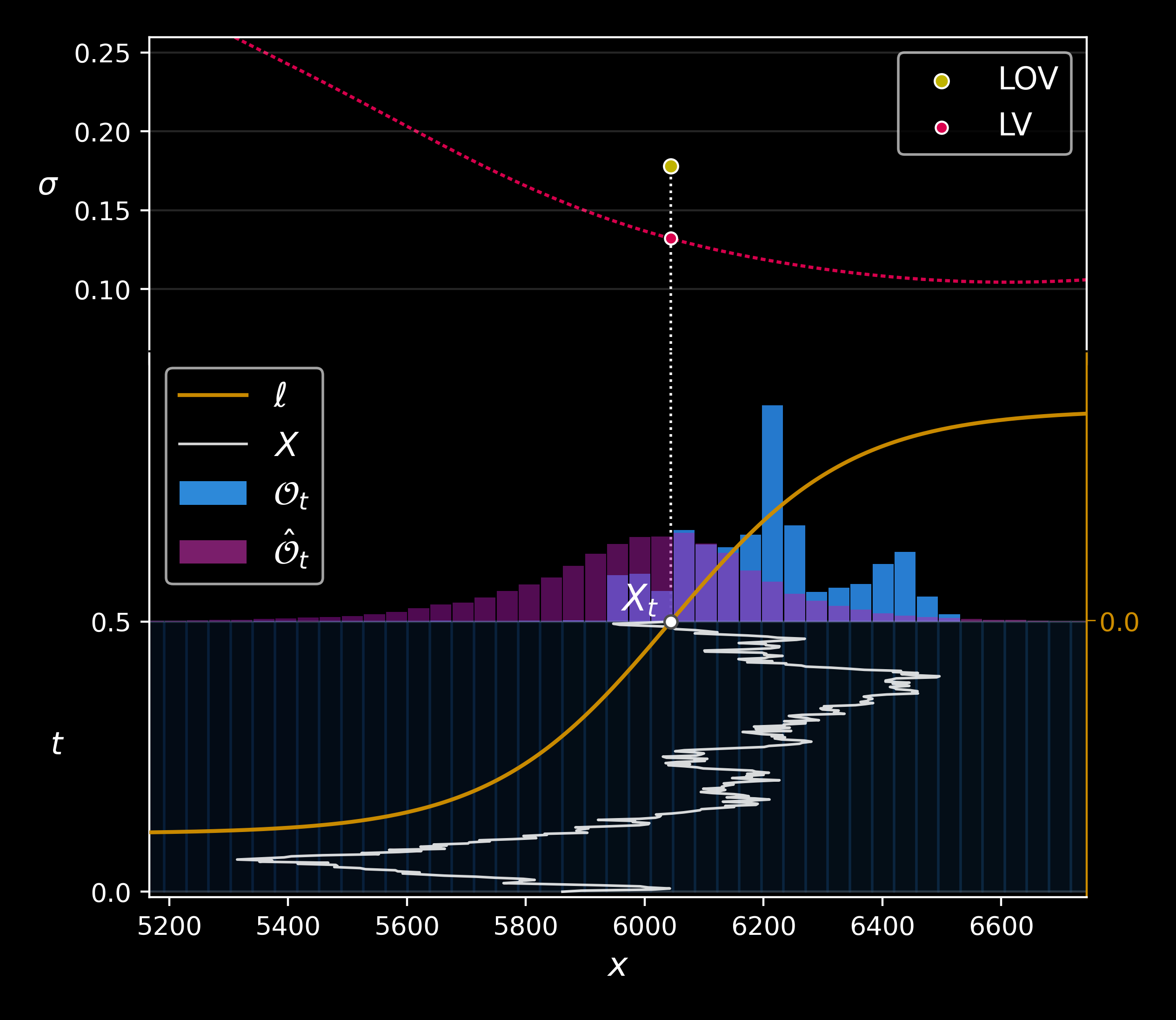}
\end{subfigure}
\begin{subfigure}[b]{0.482\textwidth}
     \centering
\includegraphics[height=2.4in,width=2.86in]{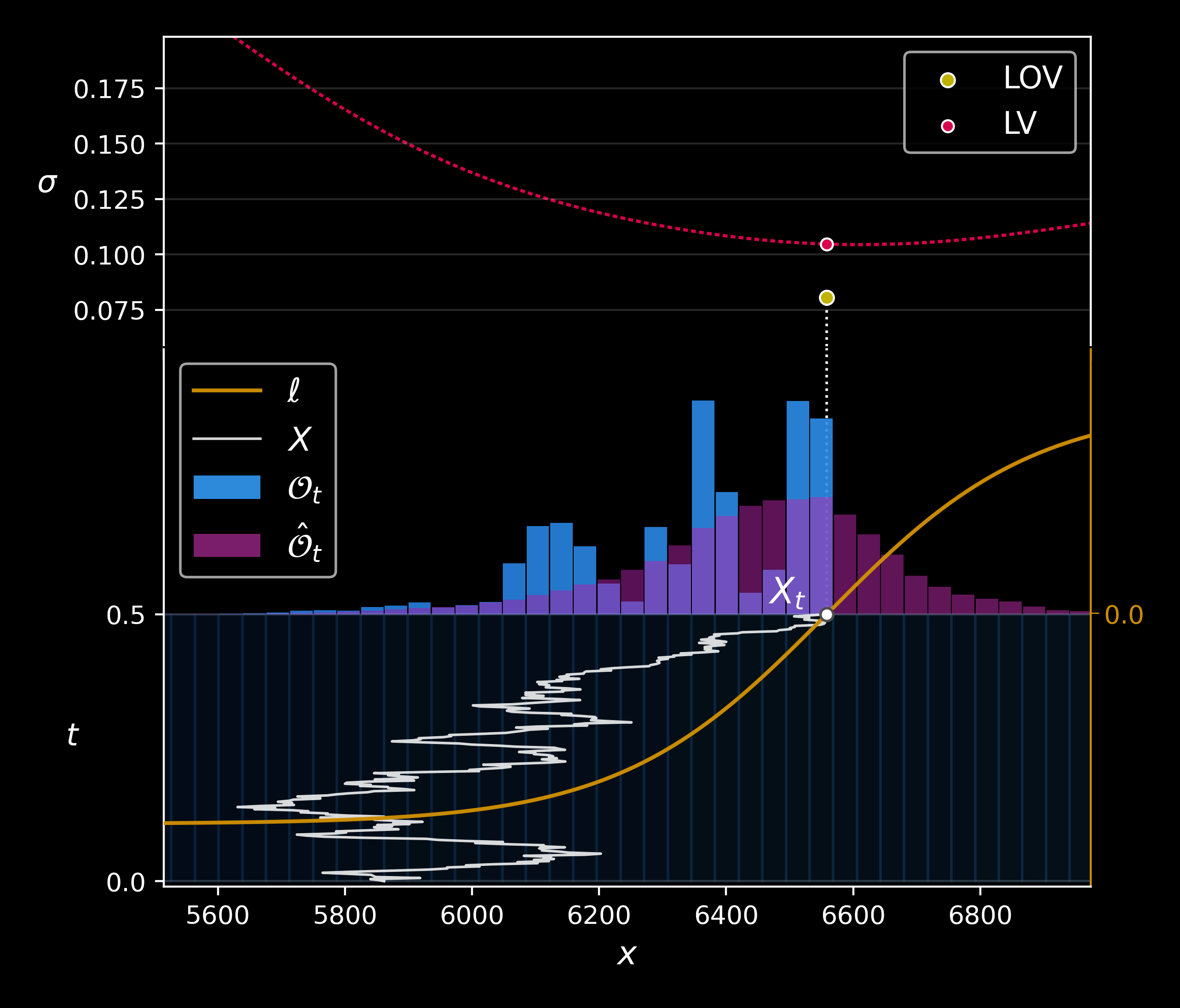}
 \end{subfigure}
\vspace{5mm}
 \label{fig:LOVSIM2}
 \end{figure}

 \begin{figure}[t]
\centering
\caption{Simulated price path and volatility in the Local Occupied Volatility (LOV) model, and corresponding local volatility (LV). }

\begin{subfigure}[b]{0.48\textwidth}
     \centering
\includegraphics[height=3in,width=2.8in]{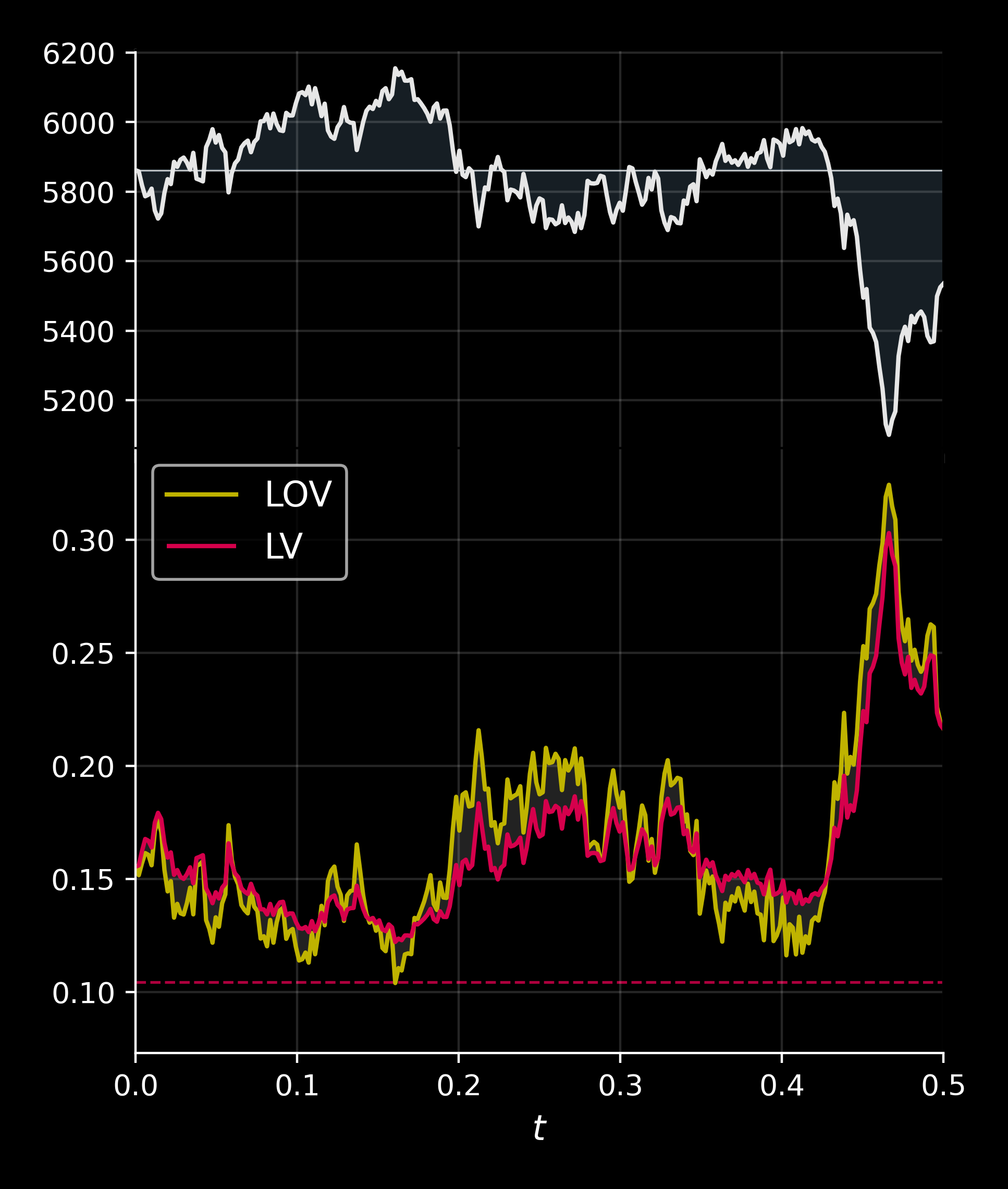}
\end{subfigure}
\begin{subfigure}[b]{0.48\textwidth}
     \centering
\includegraphics[height=3in,width=2.8in]{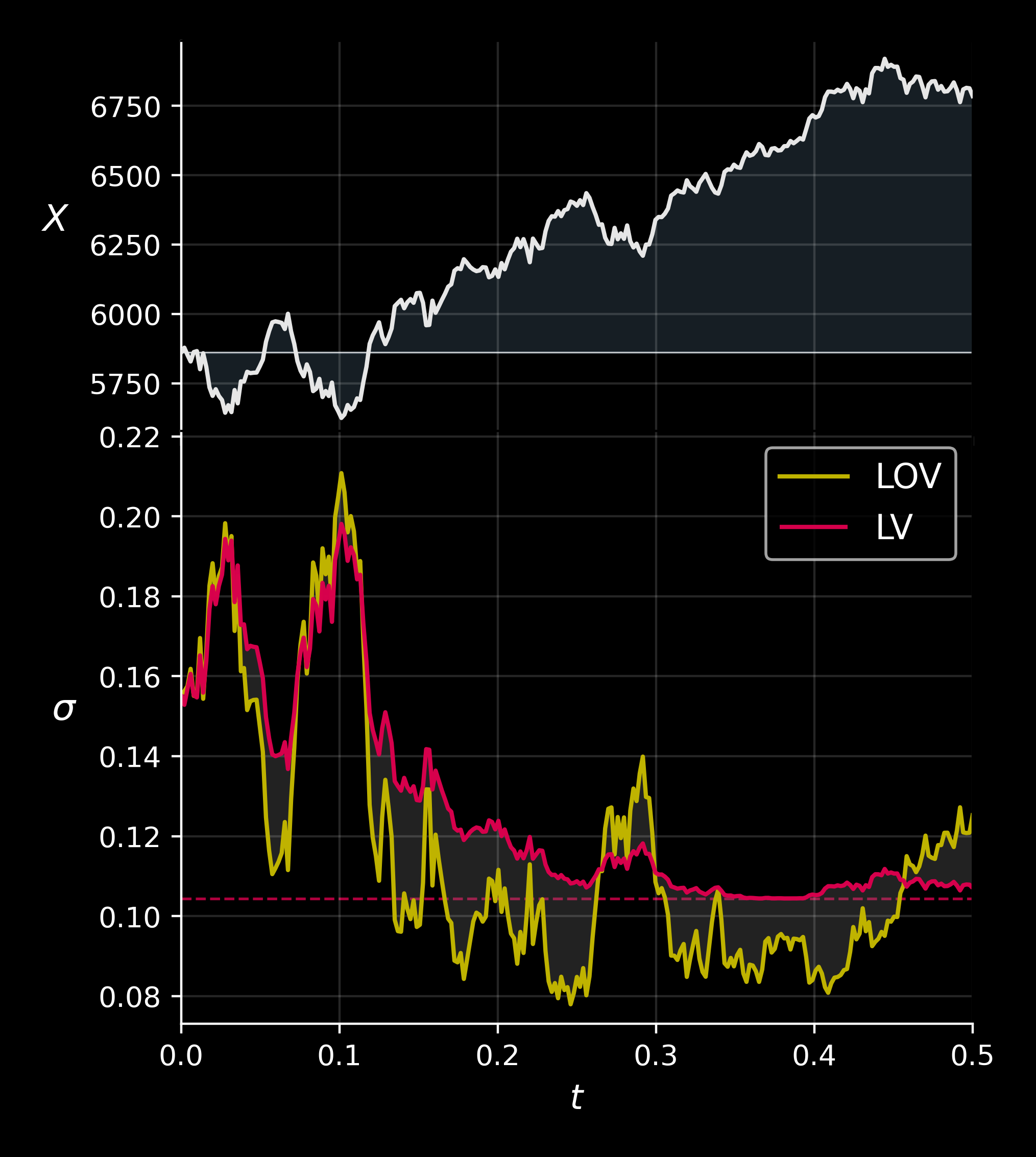}
 \end{subfigure}
 
\vspace{5mm}
 \label{fig:LOVSIM}
 \end{figure}

\subsubsection{Simulation}\label{sec:LOVSIM}
We sample trajectories from \eqref{eq:LOVX}-\eqref{eq:LOVO}  through an extension of the algorithm in \cref{app:numOSDE}. Specifically, the projected flow $\hat{\calO}$ is estimated using the particle method  \cite{GuyonHenryLabordere2012}. Concretely, the projected occupation times $\hat{O}_m(\omega_j)$ in Step III.2. is estimated by the Nadaraya-Watson estimator, 
$$\hat{O}_m(\omega_j) = \frac{\sum_{j'} \psi_{j,j'}O_m(\omega_{j'})}{\sum_{j'} \psi_{j,j'}}, \quad \psi_{j,j'} = \psi(X_{t_{n}}(\omega_j)-X_{t_{n}}(\omega_{j'})),$$
where $\psi:\R\to \R_+$ is a kernel. In our numerical experiments, we use the quartic kernel $\psi(\delta) = \frac{15}{16}(1-(\delta/h)^2)^2/h$, and set the bandwidth $h$ following the guidelines in  \cite[Chapter 11]{GuyonHL}. The  complete simulation scheme is provided in \cref{alg:LOV}. 

\begin{algorithm}[t]
\caption{(LOV Model,  Simulation) }\label{alg:LOV}

\vspace{1mm}
 $T>0$, $N \in \N$ (\# time steps),  $M$ (\# corridors), 
  $J \in \N$ (\# simulations), $t_n = n\delta t$, $\delta t = \frac{T}{N}$ (grid) 
 \\[-0.8em] 
\noindent\rule{\textwidth}{0.5pt} 
\begin{itemize}
 \setlength \itemsep{0.5em}

\item[I.] \textbf{Initialize} $O = 0 \in \R^M, X_0 = x_0$ 

\item[II.] \textbf{Generate} Gaussian random variates $Z_{n}(\omega_j) \overset{i.i.d.}{\sim}  \calN(0,1)$ 

\item[III.] \textbf{For} $n=0,\ldots,N-1$: 
\begin{enumerate}
 \setlength \itemsep{0.75em}

\item \textbf{Compute} $\hat{O}$ using the particle method:  
$$\hat{O}_m(\omega_j) = \frac{\sum_{j'} \psi_{j,j'}O_m(\omega_{j'})}{\sum_{j'} \psi_{j,j'}}, \quad \psi_{j,j'} = \psi(X_{t_{n}}(\omega_j)-X_{t_{n}}(\omega_{j'})) $$ 
    
\item \textbf{Update $\sigma$:}  
$\sigma_{t_{n}}^2(\omega_j) = \Big(\sigma_{\text{loc}}^2(t_{n},X_{t_{n}}) + \gamma_{t_{n}}^{\kappa }\sum\limits_{m=1}^M \ell(t_n,X_{t_{n}}, x_m) (O_m-\hat{O}_m)\Big)(\omega_j)$

  \item \textbf{Update $X$:} $X_{t_{n+1}}(\omega_j) = X_{t_{n}} \exp\big\{\sigma_{t_n} \sqrt{\delta t}Z_{n}  + (r-q - \frac{1}{2} \sigma_{t_n}^2)\delta t \big\}(\omega_j)$\\[-0.5em]

\item \textbf{Update $O $:}$\quad   O_m(\omega_j) \leftarrow O_m(\omega_j) + e^{\kappa t_{n}} \delta t$,  where $ \;  X_{t_{n}}(\omega_j) \in C_m$ 

\end{enumerate}
   \end{itemize}

\end{algorithm}

\begin{example} \label{ex:LOV}
Let us sample trajectories of the LOV model, where the local volatility is extracted from put/call options on the S\&P 500 Index (SPX) as of February 27, 2025. 
 Assume that $\ell(t,X_t,\cdot)$ is time-homogeneous,  and nondecreasing in $x$. For simplicity, consider  the following (scaled) hyperbolic tangent of the moneyness  $x / X_t$, 
$$\ell(t,X_t,x) = \frac{1}{4} \underline{\sigma}^2 \text{tanh}(\alpha x / X_t), \quad \underline{\sigma} = \inf_{x}\sigma_{\text{loc}}^2(x)>0, \quad \alpha >0.$$ 
According to \cref{prop:positivity}, 
the occupied variance remains positive. 
\cref{fig:LOVSIM} shows two sample paths of   the asset and  volatilities with  $\kappa = 12$ (one month window), 
and $T=0.5$ (six month). 

\cref{fig:LOVSIM2} displays the occupation measures   and sensitivity function $x\mapsto \ell(t,X_t,x)$ (yellow curve) at $t=T$, together with the local and occupied volatility. As can be seen, the local volatility is adjusted up or down according to the difference between the realized and implied occupation measures $\calO_t,\hat{\calO}_t$. In the left panel of \cref{fig:LOVSIM2}, the recent path has  spent more time above the spot than \textit{expected}, namely $\calO_t(dx) > \hat{\calO}_t(dx)$ for almost every $x>X_t$, indicating  a   possibly large drawdown resulting in higher volatility. 
Since the sensitivity function is positive above the spot and negative otherwise, the occupied volatility in the LOV model will be greater than the local volatility, further capturing the leverage effect.  The right panel of \cref{fig:LOVSIM2} shows the opposite situation, where the spot is close to the asset's recent high, which typically translates into lower volatilities.
\end{example}

%% file: Calibration.tex
\section{Joint American-European Calibration}\label{sec:nodividends}

In this section, we describe the calibration of the LOV model to market quotes for non-dividend-paying U.S. stocks. 
We then assess the model's empirical performance on the Amazon.com, Inc. (AMZN) option chain.
 While the local volatility function guarantees perfect calibration to European vanilla options, the sensitivity function $\ell$ is tuned to match the price of the American puts:
  $$ \sigma^2(\calO_t,X_t) = \underbrace{\sigma_{\text{loc}}^2(t,X_t)}_{\text{European calls}} + \ 
\gamma_t \int_{\R_+}\ \underbrace{\ell(t,X_t,x)}_{\text{American puts}} \ (\calO_t-\hat{\mathcal{O}}_t)(dx).  $$ 
The presence of American puts in the calibration process necessitates an efficient  method to price early-exercise products. This   is the focus of the next section. 



\subsection{Least Square Monte Carlo}\label{sec:LSMC}
 We extend the least square Monte Carlo (LSMC) algorithm of \cite{LS} to our setting. In short, LSMC  estimates the continuation value using regressions, and then compare it to the intrinsic value leading to nearly-optimal stopping decisions. 
In our setting, we can therefore leverage the Markov property of $(\calO,X)$ to express the continuation value as linear combination of adequate regressors  depending on the current occupation measure and spot value. We shall use the following features: 
\begin{enumerate}

    \item Black-Scholes price $p_{\text{\tiny BS}}(t,X_t,\sigma_t; K,T)$ of the $(K,T)$ put option,

    \item Laguerre polynomials of the normalized spot $X_t/x_0$, up to degree $3$,

    \item Total volatility $\sigma_t \sqrt{T-t}$ as of $t \le T$,

    \item Aggregate occupation times $\calO_t(A_n)$, with $A_n = \bigcup_{m = m_{n-1}}^{m_{n} - 1} C_m$, $m_n = 1 + [n M/5]$ ($m_0 = 1$), and interaction terms $X_t \cdot \calO_t(A_n)$, $n = 1,\ldots, 5$. 
    
\end{enumerate}

\begin{algorithm}[t]
\caption{(Neural Calibration, LOV Model) }\label{alg:calibration}

\vspace{1mm} 
 \textbf{Input:} batch size $J \in \N$, \# instruments $I \in \N$, bid/ask quotes $(p^{\text{\tiny bid}}_i, p^{\text{\tiny ask}}_i)$, option parameters $(K_i,T_i,\eta_i) \in \R_+^2\times \{\pm 1\}$, calibration weights $w_i\in \R_+$.\\[0.25em]
 \textbf{Output:} calibrated parameters $\theta \in \Theta$. 
 \\[-0.8em] 
\noindent\rule{\textwidth}{0.5pt} 
   
\begin{itemize}
 \setlength \itemsep{0.15em}
   
    \item[I.] \textbf{Generate} paths $X^{\theta}(\omega_1),\ldots, X^{\theta}(\omega_J)$ from the LOV model \eqref{eq:LOVX}-\eqref{eq:LOVO} with $\ell = \ell^{\theta}$,

    \item[II.] \textbf{Estimate} the model prices $p_i^\theta$ using Least Square Monte Carlo, 
    
    \item[III.] \textbf{Compute} the  loss function  given by   the weighted  root mean-squared error, 
    \begin{equation}\label{eq:calibLoss}
          \calL(\theta) = \Big(\frac{1}{I} \sum_{i=1}^I |w_i   (p^{\theta}_i - p^{\text{\tiny mid}}_i)|^2\Big)^{1/2}, \quad w_i = (\vartheta^{\text{\tiny mid}}_i)^{-1} \cdot \frac{1}{p^{\text{\tiny ask}}_i - p^{\text{\tiny bid}}_i}, 
    \end{equation}
    with the mid price $p^{\text{\tiny mid}}_i = \frac{p^{\text{\tiny bid}}_i + p^{\text{\tiny ask}}_i}{2}$ and Black-Scholes  Vega $\vartheta^{\text{\tiny mid}}_i$ (see \cref{rem:calibWeight}),  
    
  \item[IV.] \textbf{Update}: $  \theta \ \leftarrow \ \theta - \zeta \nabla \! \calL(\theta),   
$ with learning rate $\zeta$ given by the Adam optimizer,  

 \item[V.] \textbf{Repeat}  above steps until  calibration is complete.
\end{itemize}

\end{algorithm}

\subsection{Neural Calibration}\label{sec:calibration}

We parameterize the implied sensitivity function   by a feedforward neural network
\begin{equation}\label{eq:neuralNet}
    \ell(t,X_t,x) \approx  \ell^{\theta}(t,X_t,x) := \Phi(t,X_t,x; \theta),  \quad \Phi:[0,T] \times \R_+^2 \times \Theta \to \R_+,
\end{equation}
for some parameter set $\Theta \subseteq \R^p$, $p \gg 1$. 
The parameters are then  calibrated to market quotes of  vanilla call and put options (see \cref{alg:calibration}) using stochastic gradient descent. 
Note that a single neural network is used in \eqref{eq:neuralNet} for all time points, which presents significant benefits in terms of computational cost and smoothness in the time variable; see \cite{ReppenSonerTissot}. The optimization is also done on the full horizon all at once, often called global-in-time approach \cite{HuLauriere}.

Consider a collection of traded put and call options written on a given non-dividend paying stocks. First, we extract the implied volatility skews from the call options (being de facto European), as well as the local volatility function $\sigma_{\text{loc}}$. 
Suppose we are given  $I\in \N$ quoted put options with strike and maturity respectively denoted by  $K_i, T_i$. 
We then train the parameters of the neural sensitivity \eqref{eq:neuralNet} so that the model prices,  
$$p^{\theta}_i = \sup_{\tau \le T_i} \E^{\Q}[e^{-r\tau}(K_i - X_{\tau}^{\theta})^+],  $$
  lie between the bid and ask market prices $p^{\text{\tiny bid}}_i, p^{\text{\tiny ask}}_i$ for all $i=1,\ldots,I$. Training loop is summarized in \cref{alg:calibration}. In Step V., the calibration is deemed complete when the loss has stabilized and remains below the following  bid/ask  threshold (see also  \cite{HamidaCont})
$$\alpha = \Big(\frac{1}{I} \sum_{i=1}^I |w_i   (p^{\text{\tiny bid/ask}}_i - p^{\text{\tiny mid}}_i)|^2\Big)^{1/2} = \frac{1}{2}\Big(\frac{1}{I} \sum_{i=1}^I |w_i   (p^{\text{\tiny ask}}_i - p^{\text{\tiny bid}}_i)|^2\Big)^{1/2}.$$  



\begin{remark} \label{rem:calibWeight}(Calibration weights)
 Similar to \citet{HamidaCont},  we use the inverse  Black-Scholes  Vega, retrieved from the call options,  
 in \eqref{eq:calibLoss} to 
translate prices into implied volatilities. 
As the Vega becomes arbitrarily small for out-of-the-money options and short maturities, we also cap the weights by setting instead $(\vartheta^{\text{\tiny mid}}_i \vee \varepsilon)^{-1}$ in \eqref{eq:calibLoss} with,  e.g., $\varepsilon = 10^{-2}$; see also \cite{HamidaCont}. 
Although one could directly use the implied volatility errors $|\sigma^{\theta}_{i} - \sigma^{\text{\tiny mid}}_{i}|$ in the loss function, this would necessitate the  computation of the model implied volatilities $(\sigma^{\theta}_{i})_{i=1}^I$ at every training step, which is prohibitively costly for large $I$.
The Vega weight is then multiplied by the inverse  bid-ask spread $\frac{1}{p^{\text{\tiny ask}}_i - p^{\text{\tiny bid}}_i}$ to 
give greater importance to   liquid contracts. 
\end{remark}

\subsection{Market Data and Results}

Consider market quotes for the front five-month monthly options written on Amazon Inc. (AMZN) as of September 15, 2025.\footnote{Data retrieved from Bloomberg.} On this valuation date, the closing price of the AMZN Index was $x_0 = 231.80$. We also remove illiquid contracts with percentage bid-ask spread greater than $25\%$. This results in a calibration set of $I = 533$ instruments, consisting of 275 calls and 258 puts. The risk-free rate is fixed at $r = 3.7\%$, consistent with short-term U.S. Treasury rates as of September 15, 2025.

We sample $2^8$ pairs of antithetic paths (hence $J=2^9$) for each training iteration, which are increased to $2^{12}$ pairs after $1,000$ epochs.
In Step I., the occupied SDE is discretized using $\delta t = \frac{1}{252}$ (daily time step) and 
$M = 2^6 -1 = 63$ bins 
centered at equally spaced nodes between $x_0 \pm 2\sigma_{\text{\tiny  ATMF}}(T)\sqrt{T}$, where $\sigma_{\text{\tiny  ATMF}}(T)$ is the  at-the-money forward implied volatility for the last expiration date (February 2, 2026). We also  choose calendar  time parameter ($\kappa = 12$) to reflect a one month horizon.   The neural network  contains two hidden layers  containing $2^6 = 64$ nodes, with the ReLU and softplus activation function  for the hidden and output layers, respectively. 

Once the training is complete, all model prices are computed using the same Monte Carlo paths, consisting of $2^{12} $ antithetic pairs. The training phase and pricing take respectively about $10$ minutes and $20$ seconds on an Apple MacBook Pro M3 using CPUs only. In practice, the training can be accelerated using GPU clusters and/or using fewer instruments across the strike dimension.  

\cref{fig:FIT} illustrates AMZN put prices derived from the calibrated LOV model, with values consistently falling within or near the market bid-ask range. Furthermore, \cref{fig:LVvsLOV} provides a closer view of the strikes surrounding the forward, alongside prices from the Local Volatility (LV) model calibrated to AMZN call options. While the LV and LOV models yield nearly identical results for short-dated maturities, the LV model overestimates the value of 
the put options 
as  time to maturity increases.

\begin{figure}[H]
    \centering
    \caption{Model prices of AMZN put options compared to market bid-ask range. Six expiration dates, as of 2025-09-15.}
    
    \begin{subfigure}[t]{0.31\linewidth}
        \centering
        \includegraphics[width=\linewidth, height=1.85in, keepaspectratio=false]{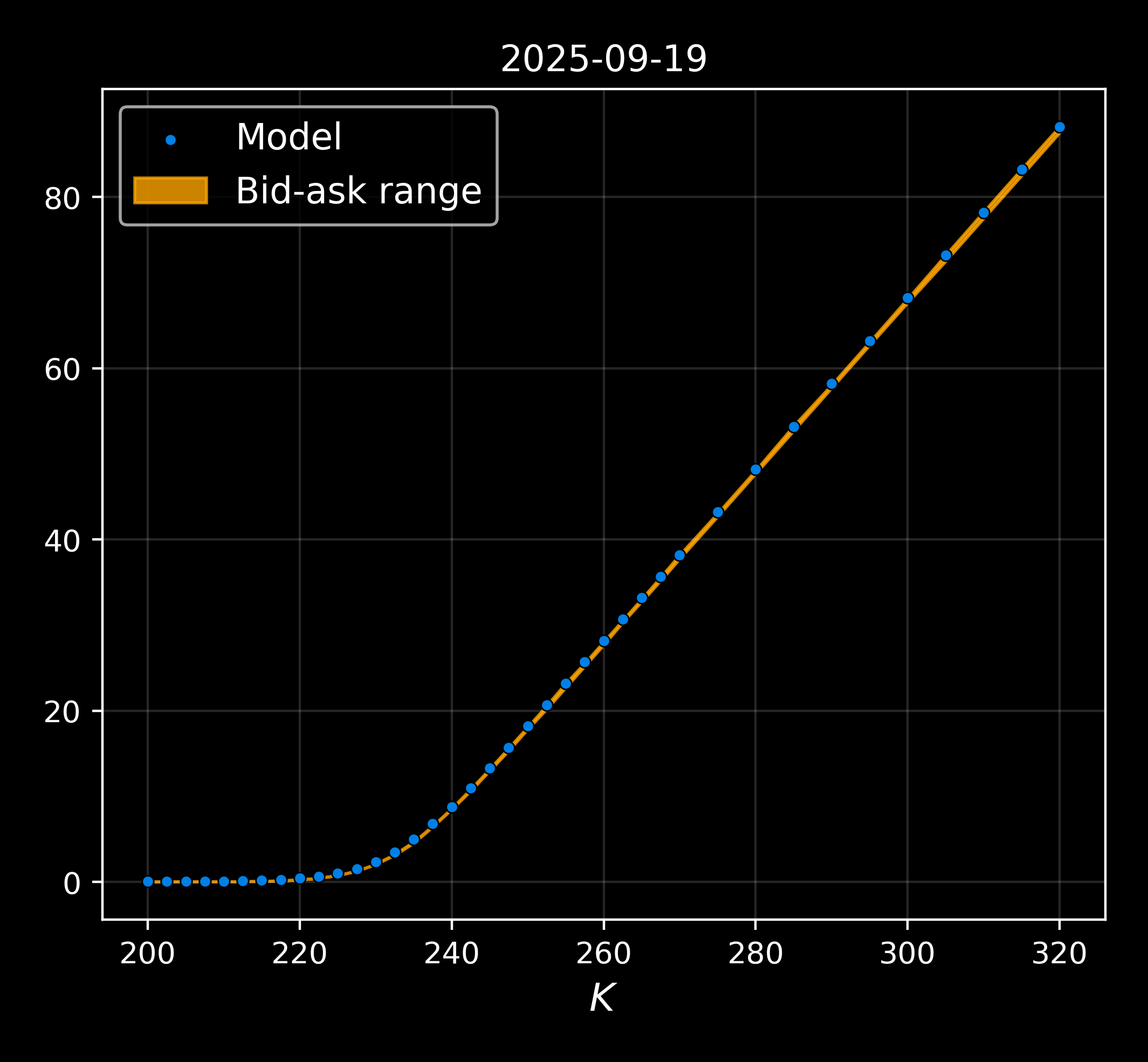}
        \label{fig:fit1}
    \end{subfigure}\hfill
    \begin{subfigure}[t]{0.31\linewidth}
        \centering
        \includegraphics[width=\linewidth, height=1.85in, keepaspectratio=false]{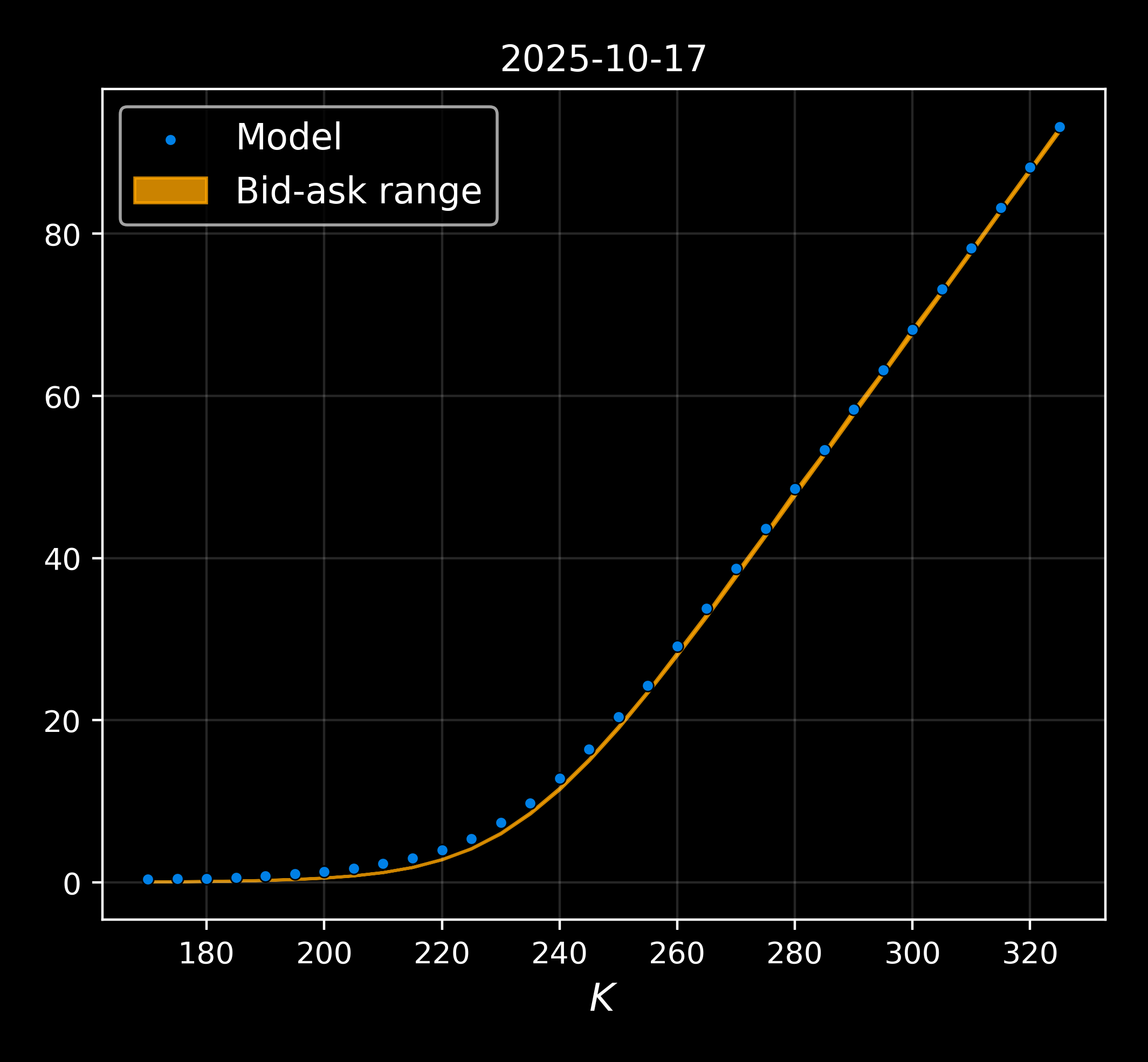}
        \label{fig:fit2}
    \end{subfigure}\hfill
    \begin{subfigure}[t]{0.31\linewidth}
        \centering
        \includegraphics[width=\linewidth, height=1.85in, keepaspectratio=false]{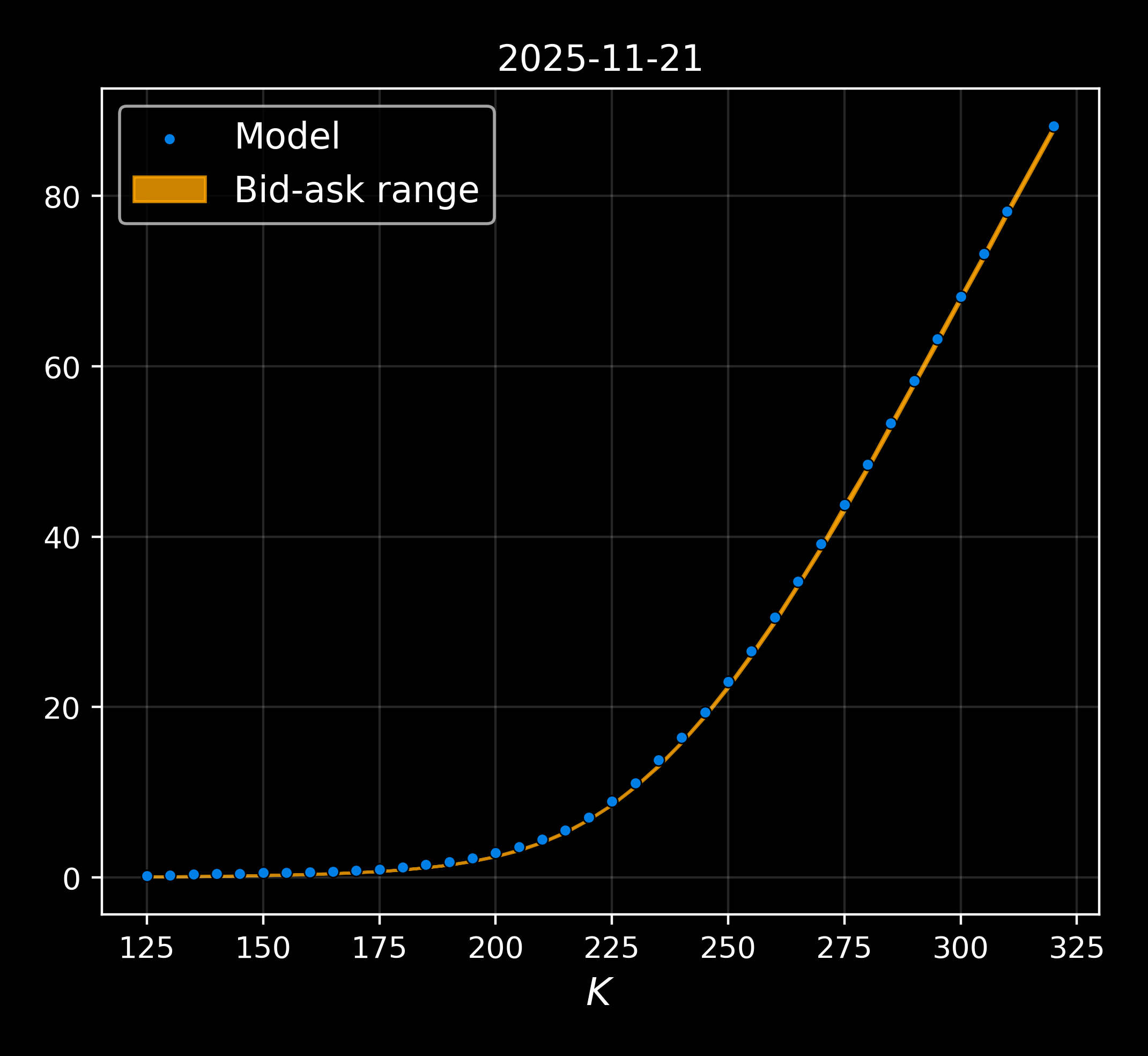}
        \label{fig:fit3}
    \end{subfigure}

    \medskip 

    \begin{subfigure}[t]{0.31\linewidth}
        \centering
        \includegraphics[width=\linewidth, height=1.85in, keepaspectratio=false]{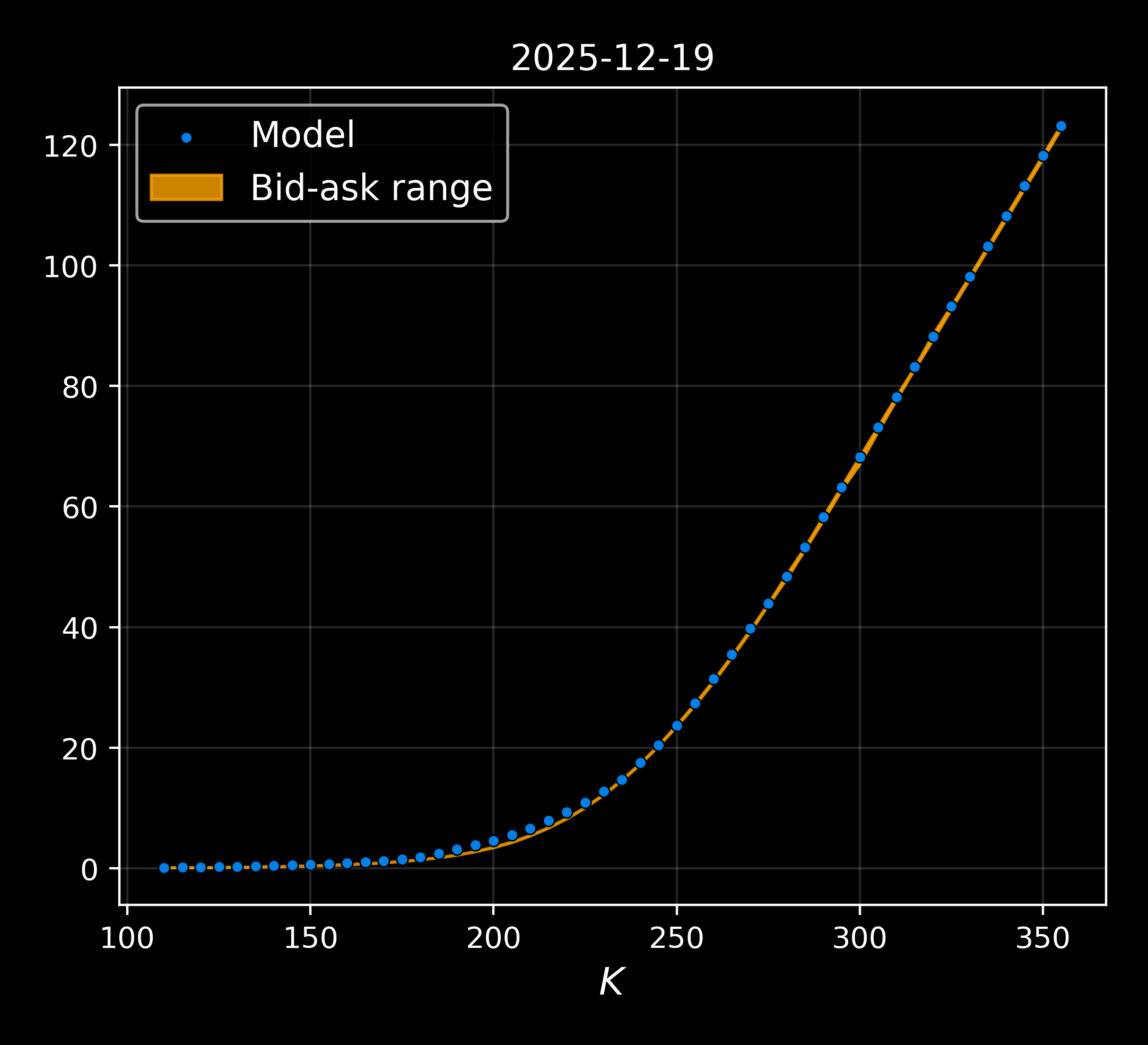}
        \label{fig:fit4}
    \end{subfigure}\hfill
    \begin{subfigure}[t]{0.31\linewidth}
        \centering
        \includegraphics[width=\linewidth, height=1.85in, keepaspectratio=false]{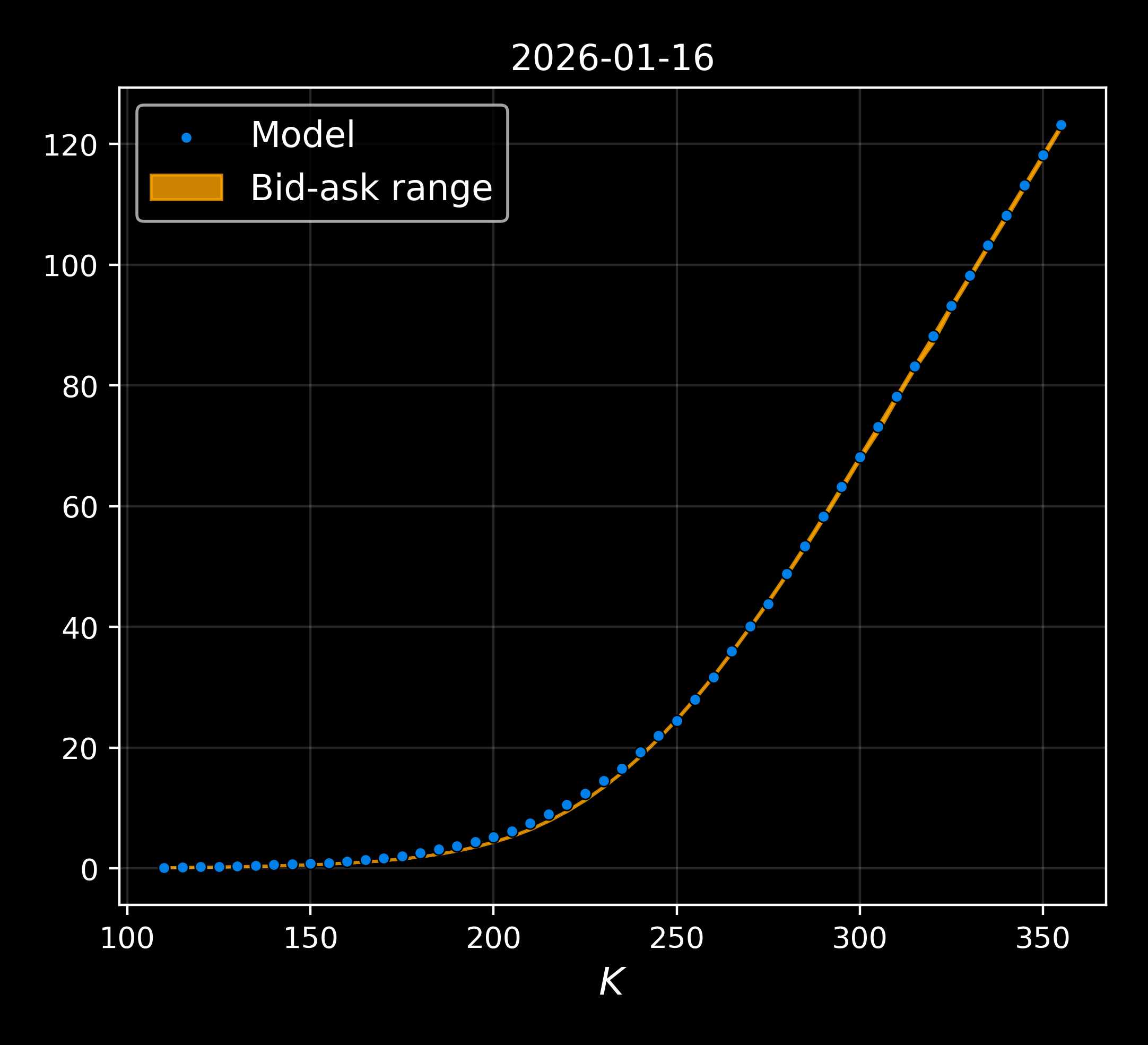}
        \label{fig:fit5}
    \end{subfigure}\hfill
    \begin{subfigure}[t]{0.31\linewidth}
        \centering
        \includegraphics[width=\linewidth, height=1.85in, keepaspectratio=false]{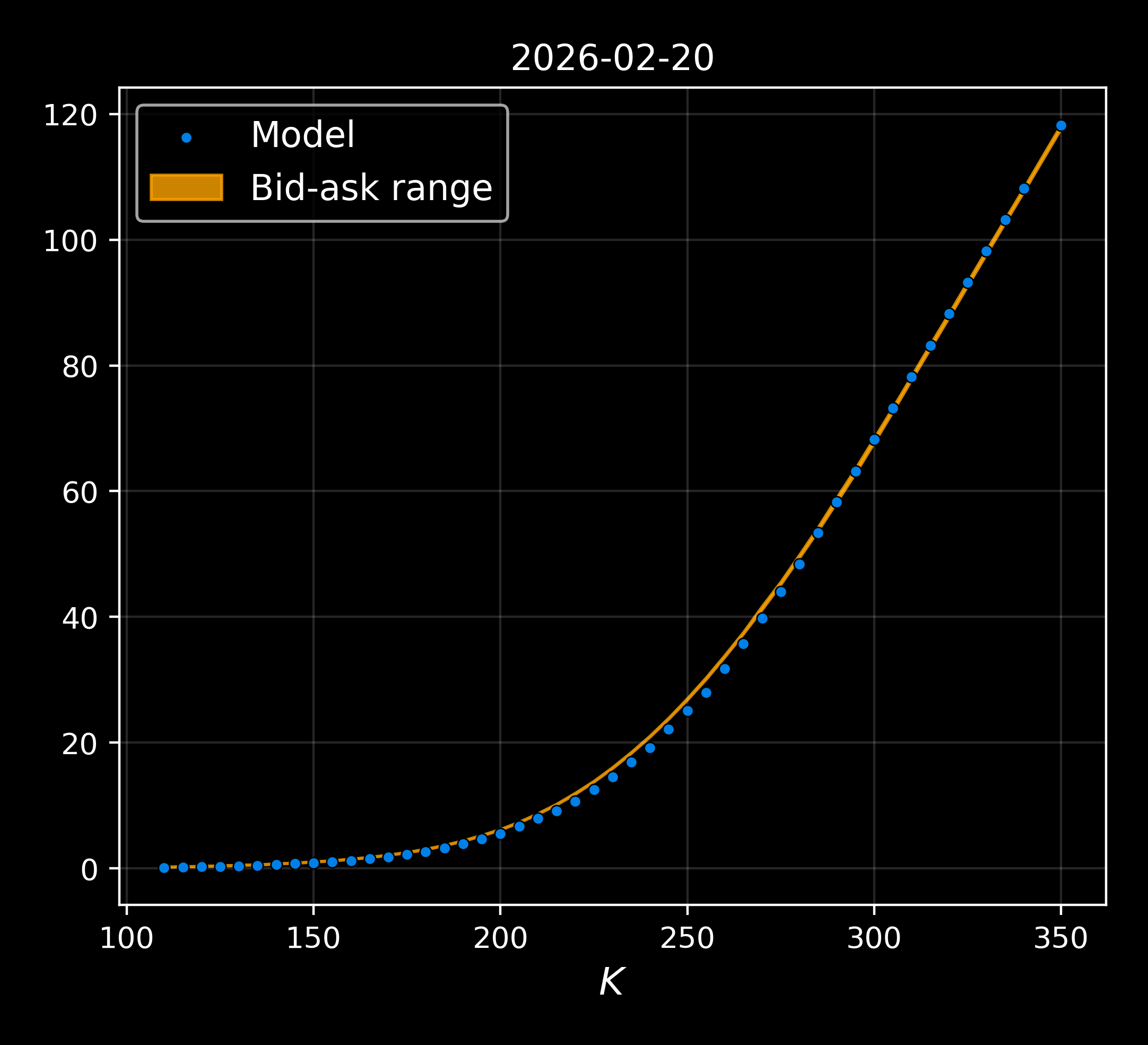}
        \label{fig:fit6}
    \end{subfigure}
    \label{fig:FIT}
\end{figure}

\begin{figure}[H]
    \centering
    \caption{AMZN put option prices in the LOV model,  compared to the local volatility (LV) model and market bid-ask range as of 2025-09-15. }
    
    \begin{subfigure}[t]{0.31\linewidth}
        \centering
        \includegraphics[width=\linewidth, height=1.85in, keepaspectratio=false]{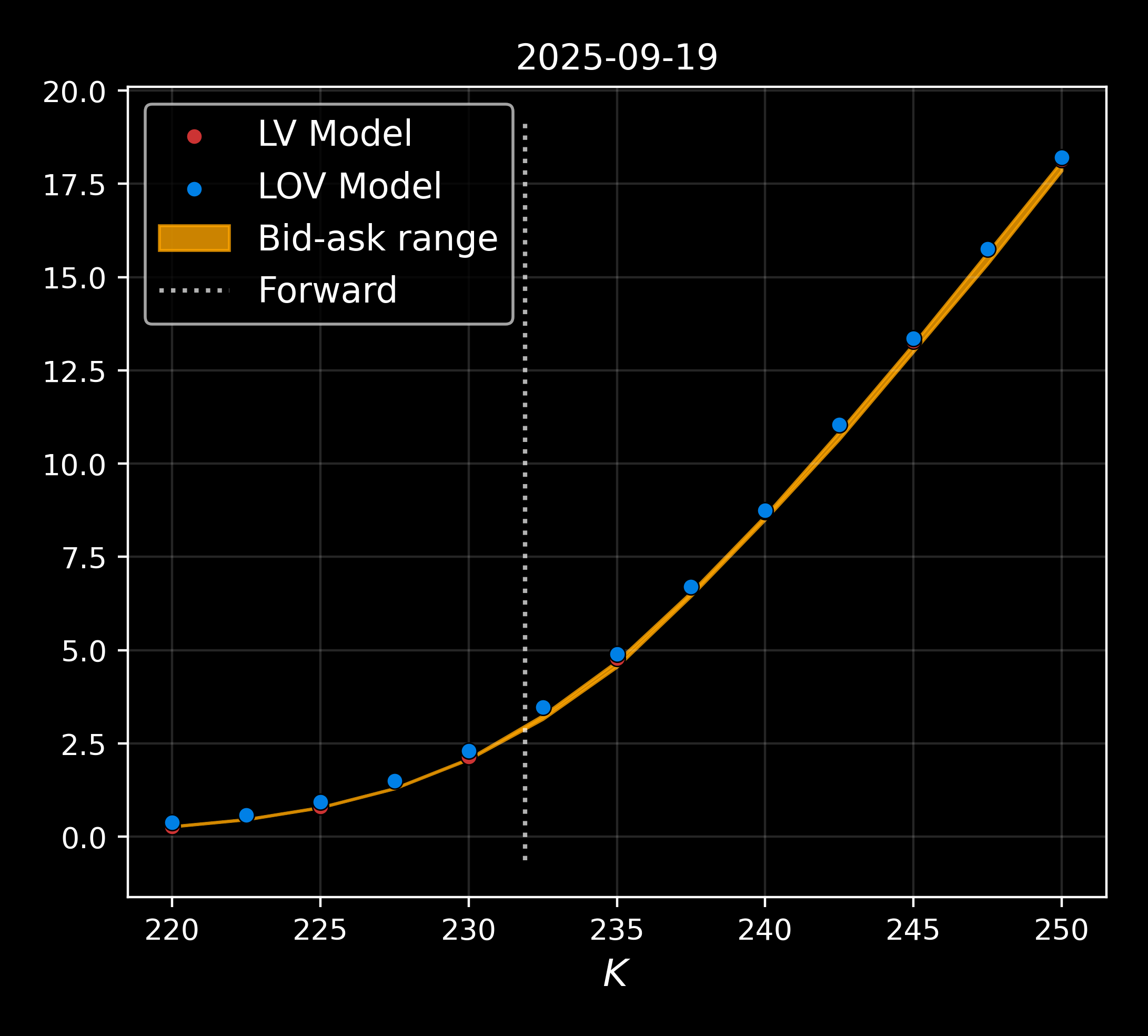}
        \label{fig:fit1}
    \end{subfigure}\hfill
    \begin{subfigure}[t]{0.31\linewidth}
        \centering
        \includegraphics[width=\linewidth, height=1.85in, keepaspectratio=false]{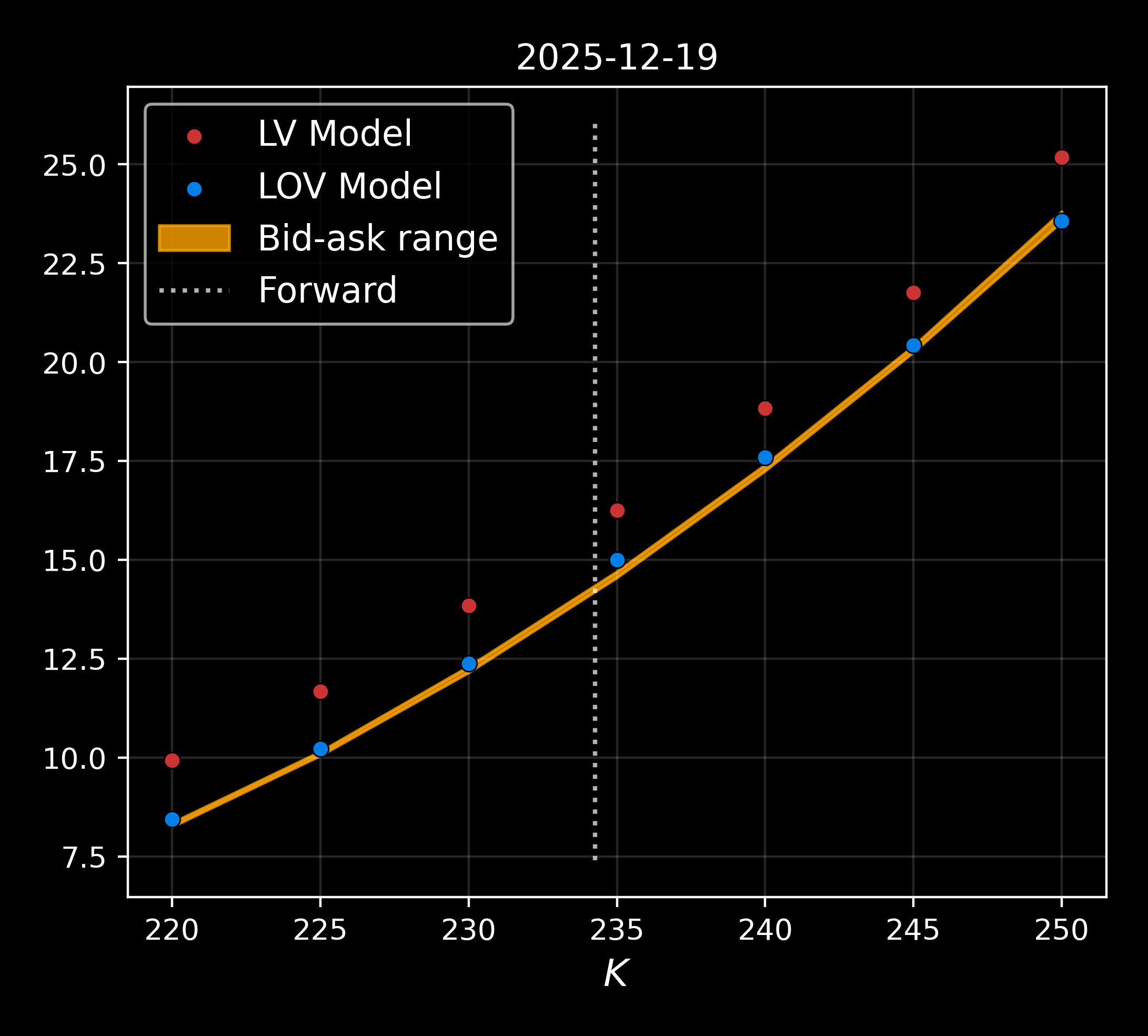}
        \label{fig:fit2}
    \end{subfigure}\hfill
    \begin{subfigure}[t]{0.31\linewidth}
        \centering
        \includegraphics[width=\linewidth, height=1.85in, keepaspectratio=false]{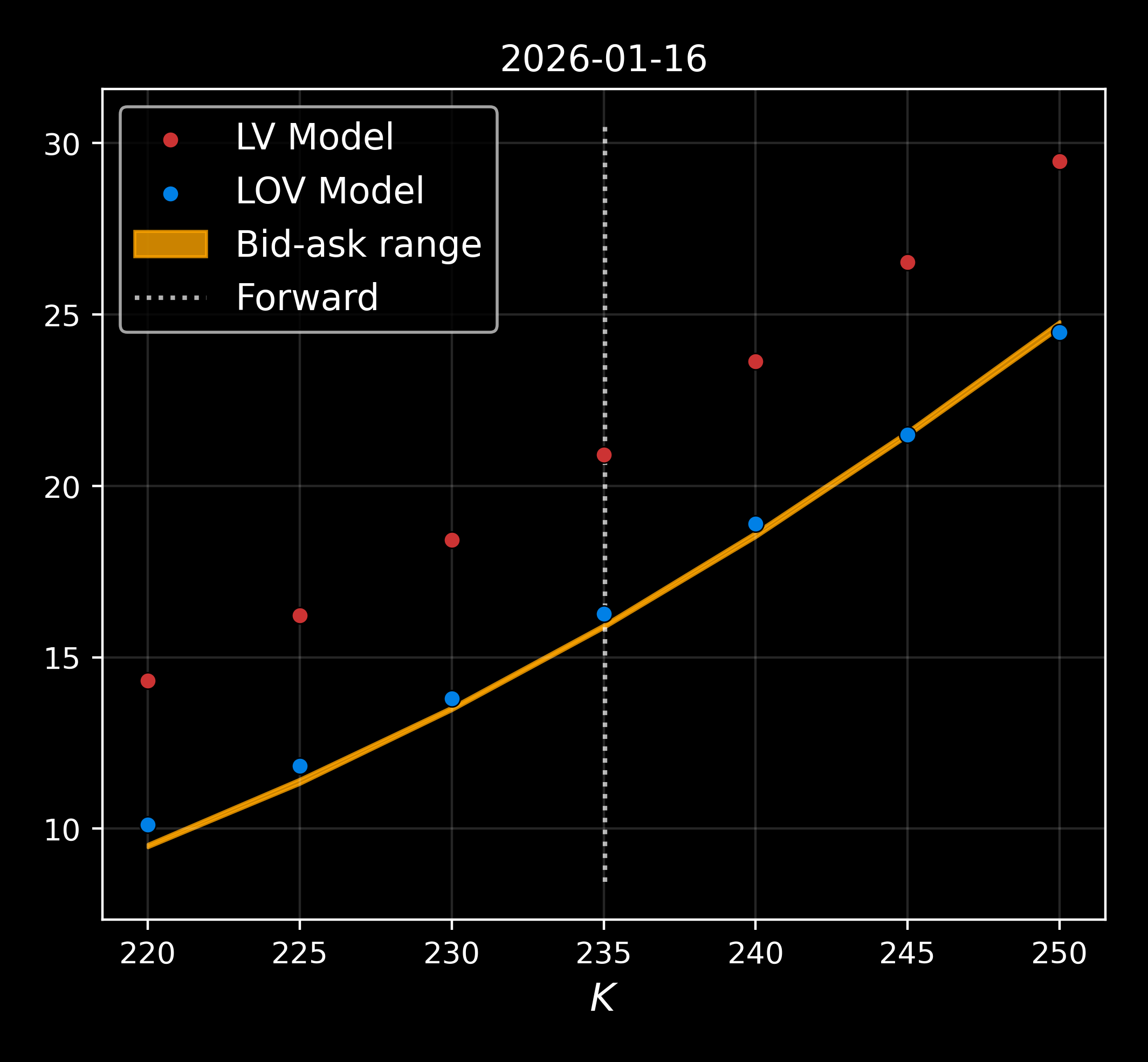}
        \label{fig:fit3}
    \end{subfigure}
\label{fig:LVvsLOV}
\end{figure}

\cref{fig:calibtratedEll} shows the neural sensitivity  $x\mapsto \ell^{\theta}(t,X_t,x)$ for $t  = 1/12$ (one month) and various values for the spot $X_t$.  First, the sensitivity is found to be non-increasing in $x$, a discrepancy with the increasing behavior anticipated under the historical measure to reflect the leverage effect. 
This phenomenon has been consistently observed over multiple calibration runs with varied neural network initializations. Also,  observe the upward  shift of $\ell$ as the spot $X_t$ increases. 
Finally, the calibrated sensitivity turns out to have little dependence with respect to the time variable. This indicates that using time-homogeneous sensitivity functions is enough in this context.  

 \begin{figure}[H]
\centering
\caption{Calibrated neural sensitivity function $x\mapsto \ell^{\theta}(t,X_t,\cdot)$ for $t  = 1/12$ (one month) and various values for the spot $X_t$.  }

\centering
 \includegraphics[height=2.2in,width=3in]{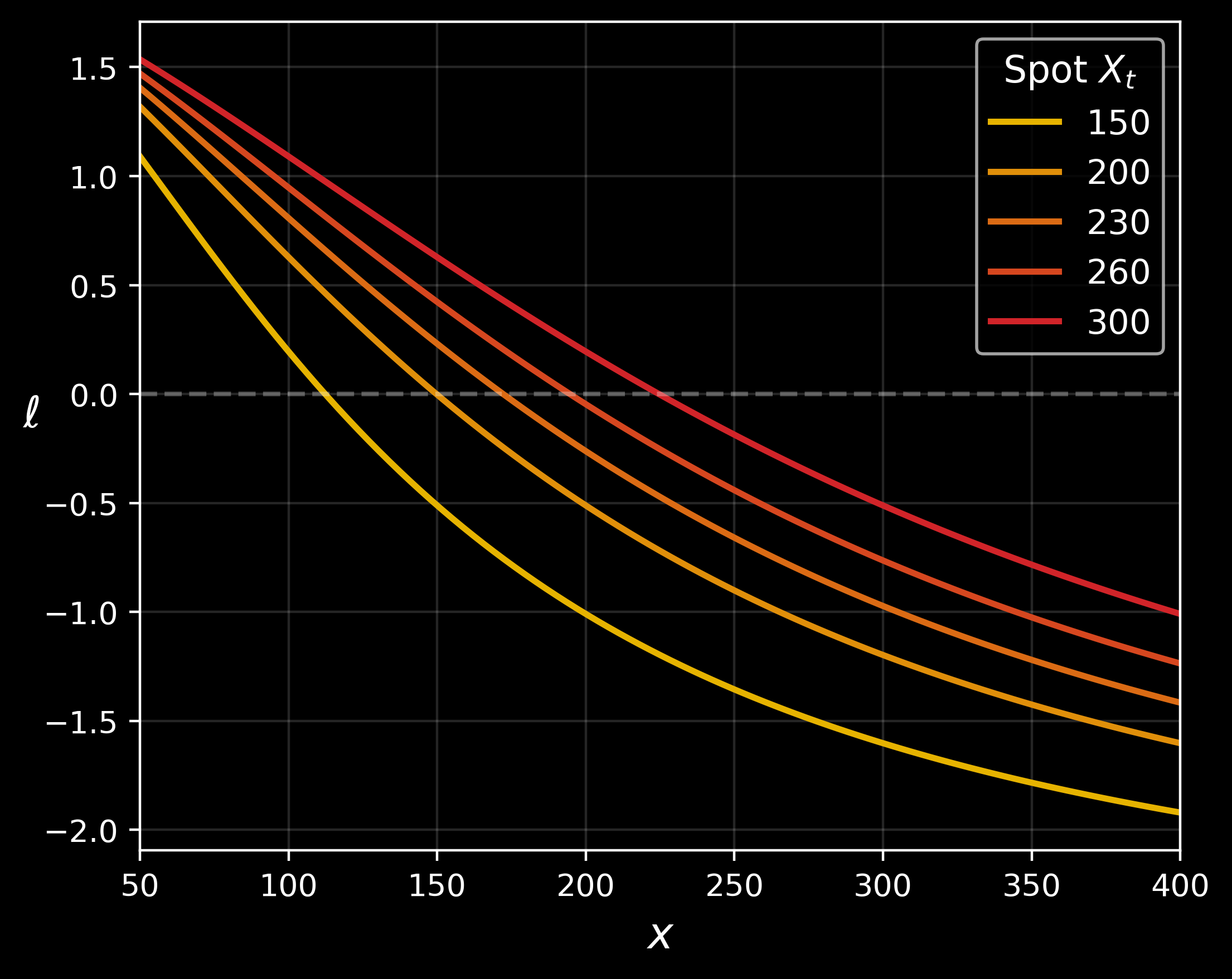}
\label{fig:calibtratedEll}
\end{figure}

%% file: Conclusion.tex
 \section{Adding Dividends} \label{sec:conclusion}


When the underlying asset issues dividends, both call and put options carry an early exercise premium. 
This implies that matching call option quotes is no longer equivalent to  the smile calibration condition \eqref{eq:smileCalibration}. 
Hence, there is no need to consider  McKean dynamics in this case, so the instantaneous variance in the LOV model becomes
$$\sigma(\calO_t,X_t)^2 = \sigma_{\text{loc}}(t,X_t) + \gamma_t^{\kappa} \int_{\R_+}\ \ell^{\Q}(t,X_t,x)  \hspace{0.2mm} \calO_t(dx). $$
Note that the local volatility function is \textit{not} obtained from Dupire's formula. 
Instead, $\sigma_{\text{loc}}$ shall be   parameterized, e.g. using neural networks,  along the sensitivity function $\ell$ to quotes of American put and call options. A notable challenge in this context is the discrete nature of dividends, making the valuation of American options much more involved. We thus leave this direction for future work. 

An alternative approach, called \textit{de-Americanization}, consists of extracting  the prices of synthetic European options from observed American options. This can be done using a binomial model on a fine grid, or using neural networks \cite{LindGatheral2025}. One can then calibrate a volatility model to the synthetic European prices or their corresponding implied volatility surface. 
While the calibration to European-style options  is significantly faster, the de-Americanization step is  model-dependent and thus prone to error \cite{burkovska2018calibration}. Crucially, there is no guarantee that a model calibrated to the synthetic European prices will accurately recover the original American option quotes. 